\documentclass[a4paper]{scrartcl}

\pdfoutput=1

\usepackage{amsfonts}
\usepackage{amsmath}
\allowdisplaybreaks[1]
\usepackage{amsthm}
\usepackage{amssymb}
\usepackage{graphicx}
\usepackage{color}
\usepackage[linktoc=all, colorlinks=true]{hyperref}

\newcommand{\interior}{\ensuremath{\textrm{int}}}
\newcommand{\R}{\ensuremath{\mathbb{R}}}
\newcommand{\N}{\ensuremath{\mathbb{N}}}
\newcommand{\Sph}{\ensuremath{\mathbb{S}}}
\newcommand{\Atlas}[1]{\ensuremath{\mathcal{A}\left({#1}\right)}}

\newcommand{\dom}{\ensuremath{\textrm{dom}}}
\newcommand{\ran}{\ensuremath{\textrm{ran}}}

\newcommand{\BP}[1]{\ensuremath{\textrm{BP}\left({#1}\right)}}
\newcommand{\BPM}[1]{\ensuremath{\textrm{B}_{\textrm{ch}}\left({#1}\right)}}
\newcommand{\ABBPM}[1]{\ensuremath{\mathcal{B}_{\textrm{ch}}\left({#1}\right)}}
\newcommand{\EX}[1]{\ensuremath{\textrm{EX}\left({#1}\right)}}

\newcommand{\rest}[2]{\ensuremath{{#1}\circ{#2}}}

\newcommand{\cur}[1]{\ensuremath{\mathcal{#1}}}
\newcommand{\covers}{\ensuremath{\rhd}}
\newcommand{\contact}{\ensuremath{\perp}}

\newtheorem{theorem}{Theorem}[section]

\newtheorem{definition}[theorem]{Definition}
\newtheorem{lemma}[theorem]{Lemma}
\newtheorem{corollary}[theorem]{Corollary}
\newtheorem{proposition}[theorem]{Proposition}
\theoremstyle{plain}
\newtheorem{example}[theorem]{Example}

\begin{document}

\title{The chart based approach to studying the global 
structure of a spacetime 
induces a coordinate invariant boundary.}

\author{B. E. Whale\hspace{1pt}\textsuperscript{1}}

\maketitle

\footnotetext[1]{Department of Mathematics 
and Statistics, University
of Otago, P.O. Box 56, Dunedin 9054, New Zealand\\
Tel.:+64-3-479-7758,
Fax.:+64-3-479-8427\\
\texttt{bwhale@maths.otago.ac.nz}
}
\begin{abstract}
	I demonstrate that
  the chart based approach to the study of the global structure of
  Lorentzian manifolds induces a homeomorphism of the manifold 
  into a topological space as an open dense set. 
  The topological boundary of 
  this homeomorphism is a
  chart 
  independent boundary of ideal points equipped with
  a topological structure and a
  physically motivated classification.
  I show that this new boundary contains all other boundaries that can
  be presented as the topological boundary
  of an envelopment. Hence, in particular, it is a generalisation of
  Penrose's conformal boundary.
  I provide three detailed examples:
  the conformal compactification
  of Minkowski spacetime, Scott and Szekeres' analysis of the Curzon
  singularity and Beyer and Hennig's analysis of smooth Gowdy 
  symmetric generalised Taub-NUT spacetimes.   
\end{abstract}

\tableofcontents

\section{Introduction}\label{sec.introduction}

	When presented with a spacetime and asked what its global structure is the
	usual approach is to construct charts that
	`extend to the edge of the manifold' and then analyse the behaviour of 
	coordinate invariant quantities
	in the charts, e.g.\ the Kretschmann scalar or a 
  congruence of geodesics.
	The charts are then altered so that the invariant quantities 
	have suitably nice behaviour in the limit to `the edge of
  the manifold'. The construction
	of the chart $\beta$, 
  described in Section \ref{sec:curzon}, is an example
	of this.
	The structure of these altered charts `on the edge' is then taken to be a 
	boundary
	for the spacetime, with justification coming from the nice behaviour of 
	the particular invariants used.
	Minkowski spacetime, equipped with Penrose's conformal boundary, viewed
	as a collection of coordinate transformations, is a
	well known example, see Section \ref{sec:Minkowski}.
	The original justification was based on the
	conformal invariance of the zero rest-mass free-field equation for any
	spin value, \cite{citeulike:8978534,citeulike:9600515,citeulike:10075981}.
	
	Unfortunately this approach is not necessarily chart independent. Could other 
	charts that
	also allow for nice behaviour of the invariants `on the edge' exist? 
	Could these 
	charts give boundaries that have different topologies? Given a
	chart that `extends to the edge' of a manifold and 
  an invariant which has nice behaviour
	`on the edge' do there exist other invariants
	which do not have nice behaviour `on the edge'? One of
	the motivators in the field of boundary constructions is the desire for a
	boundary that generalizes Penrose's conformal boundary and avoids these
	coordinate dependence problems, \cite{Geroch1968Local}.
  See \cite[Section 6]{citeulike:12270518} for further discussion of this
  and a review of boundary constructions in general.

	Penrose's conformal boundary, while well adapted to describing causal
	structure, fails to describe the gravitational field of an asymptotically
	flat space time which has non-zero ADM mass, 
	\cite[page 569]{citeulike:10075981}. The gravitational field is singular
	at both spacelike and past / future timelike infinities, 
	\cite[pages 568--569]{citeulike:10075981}. In the context
	of the discussion above: the conformal boundary is given by charts 
	for which the 
	invariant causal structure has nice behaviour `on the edge'.
  The same charts, however, fail to provide
	nice behaviour for the metric `on the edge', when
	the spacetime has non-zero ADM mass.
	This is one of the reasons underlying  Friedrich's modifications
	of Penrose's conformal boundary, \cite{citeulike:8879192,citeulike:8973162}.
	
	The definition of a coordinate
  independent boundary, that is easy to construct, fits with physical
	intuition and generalises some appropriate properties of
  Penrose's conformal
  boundary, has proven to be very difficult, 
	\cite[Section 2]{Ashley2002a}, \cite[Section 2.3]{Whale2010} and 
  \cite{citeulike:12270518}.
	The literature 
	contains numerous suggested
	boundaries, examples of how they behave badly and their subsequent 
	alterations, \cite[Section 2]{Ashley2002a} and \cite[Section 2.3]{Whale2010}. 
	As a result the boundaries that are still being developed
	are sophisticated and 
	technical, 
	\cite{citeulike:12270518,ScottSzekeres1994,Flores2010Final,%
  citeulike:8664207}. 

	The apparently necessary technicalities of 
	boundary constructions divorces them from the more natural 
  chart based approach. 
	As a consequence of this the field of boundary constructions
	has become separated from the main stream of research in Relativity.
	Penrose's conformal boundary is a notable exception. It has been successfully 
	applied in many areas, \cite{citeulike:12270518}.
	This paper is motivated by two of the reasons why; (i) it
	can be constructed using charts and (ii) it can be presented as
  the topological boundary of an envelopment of the manifold.
  This gives the boundary a convenient 
	construction and a differential structure. In particular, calculations
  on the conformal boundary can be performed using charts in some larger
  manifold. Coordinate independence is then ensured by requiring that
  quantities defined on the boundary transform in the appropriate way
  with respect to the charts of the larger manifold.

  In this paper I attempt to reconcile this divorce by showing that
  the chart based approach to the study of global structure equips
  the manifold with a chart independent boundary which; (i)
  contains (and therefore is a generalisation of) Penrose's conformal
  boundary, see Section \ref{sec.compatiblecharts} and in particular
  Proposition \ref{prop.Q(M)fromphi},
  (ii) has a coordinate description, see Section \ref{sec.cicib} and the
  discussion following
  Definition \ref{def.quotientmap},
  and (iii)
  allows for the use of differential structure on the boundary, see
  for example Proposition \ref{lem:embedding given by an extension}
  and the discussion following Definition \ref{def.quotientmap}.

  The boundary is based on the following idea. Each chart 
  $\alpha:U\subset M\to \R^n$, with a suitable
  assumption on its behaviour, Definition \ref{def:extension},
  can be associated with some larger subset, $V$, of $\R^n$
  so that $\alpha(U)\subset V$.
  The topological boundary of $\alpha(U)$ relative 
  to $V$, $\partial_V\alpha(U)$, can then be considered
  as a local representation of ``the boundary of $M$''.
  Because of the assumption on $\alpha$,
  the set $V$ can be interpreted as the domain of
  some larger chart in some enveloping manifold of $M$, see 
  Proposition \ref{lem:embedding given by an extension}.
  This equips the local representation of the boundary of $M$,
  $\partial_V\alpha(U)$,
  with a differential structure.
  These details are described in Section \ref{sec.chart_boundary}
  and the discussion after Definition \ref{def.quotientmap}.

  Calculations performed in these charts can now be said to describe
  coordinate dependent quantities on local representations of
  the boundary of $M$.
  If the quantity transforms the right way under 
  ``coordinate transformations on the boundary'' then it is possible
  to claim that the coordinate dependent quantities describe, with respect
  to the chosen chart, a coordinate independent object on the boundary.
  This idea is explored in Sections \ref{sec.relbetweenbound} and
  \ref{sec.cicib} where a topological space is constructed,
  Definition \ref{def.quotientmap} (which gives
  a concrete realisation of the boundary of $M$), and a precise
  definition of the coordinate transformations on the boundary
  is given, Proposition \ref{prop.coordtransonboundarydef}.

  The topological boundary, $\ABBPM{M}$, of $M$ in this topological space is
  the chart independent object that the local representations,
  $\partial_V\alpha(U)$,
  represent. Given some quantity that is well defined on some local 
  representation of $\ABBPM{M}$, it is very likely that there will exist
  some other local representation in which that quantity is not
  well defined. This can be thought of as a manifestation of
  the same issue that the conformal boundary has with gravitational 
  fields of spacetimes with non-zero ADM mass. 
  However, the freedom to choose local representations of the boundary that
  are appropriate for the quantity being studied allow for $\ABBPM{M}$
  to cope with this issue.

  This is such an important point that a digression is warranted.
  The construction of $\ABBPM{M}$ makes only one
  topological restriction (Definition \ref{def:extension})
  on the charts used. In particular, no geometric or physical
  information is used. This ensures
  that $\ABBPM{M}$ is uniquely associated to any manifold and frees it
  from claims of arbitrariness. The downside is that, in general,
  explicit construction of $\ABBPM{M}$
  can never be completed. This is, however, not needed as calculations
  on local representations of $\ABBPM{M}$, once checked for
  coordinate independence on $\ABBPM{M}$, are sufficient to produce
  well defined quantities on $\ABBPM{M}$. In this way the standard chart
  based approach to the study of global structure can be freed from
  claims of coordinate dependence. Because the construction
  of $\ABBPM{M}$ makes only 
  a mild
  topological restriction on charts this boundary
  provides no guidance on which charts are suitable
  for the calculation of particular quantities. This is a point
  that I shall discuss in more detail below and in 
  Section \ref{sec:Examples}.

  Section \ref{sec.relbetweenbound} defines 
  an
  equivalence relation, Definition \ref{def.equiv},
  which describes how the different local representations of the 
  boundary of $M$ are related to each other.
  Given two charts, $\alpha:U\subset M\to \R^n$ and 
  $\beta:V\subset M\to \R^n$
  and $x\in U\cap V$ then $x$ has the two representations
  $p=\alpha(x)\in \R^n$ and $q=\beta(x)\in\R^n$. We know that
  $p$ and $q$ are two different representations of
  some third, chart independent, point $x$.
  If we did not have the object $x$, but only the functions
  $\beta$ and $\alpha$ as well as the points $p$ and $q$, then
  it is still possible to determine if $p$ and $q$ represent
  some third, chart independent, point by asking if
  $\beta\circ\alpha^{-1}(p)=q$. This defines an equivalence
  relation on the collection of all charts and points in the
  images of the charts. The extension of this equivalence relation to 
  local representations of the boundary of
  $M$ is exactly the equivalence relation
  given in Definition \ref{def.equiv}.

  With the equivalence relation it is possible to construct a topological
  space into which the manifold is embedded as an open dense set, 
  see Section \ref{sec.cicib}. Continuing the analogy described above,
  if $p\in\alpha(U)$ and $q\in\beta(V)$ are equivalent then the
  third, chart independent, point $x$ can be described
  as the equivalence class, $[(\alpha, p)]$, of $(\alpha, p)$
  such that $(\beta, q)\in[(\alpha, p)]$ if and only if
  $\beta\circ\alpha^{-1}(p)=q$. The topological space
  is composed of these equivalence classes for the extension of the
  equivalence relation to the local representations of the boundary.
  It is in this sense that the topological space is coordinate independent.
  The definition of the topological space,
  Definitions \ref{def.quotientmap} and \ref{def:chart_boundary},
  is as a quotient
  space of the disjoin union of $\alpha(U)\cup\partial_V\alpha(U)$
  over all allowable charts $\alpha$, see Definition 
  \ref{def:extension}.
  The quotient map has several nice properties, see the discussion
  after Definition \ref{def.quotientmap} and Appendix
  \ref{sec.topQ(M)}. In particular
  the quotient map, restricted to $\alpha(U)\cup\partial_V\alpha(U)$
  for any allowable chart, is a homeomorphism.
  This implies that the boundary, $\ABBPM{M}$, of $M$ has a local description
  in terms of charts,
  mediated by these homeomorphisms,
  so that the induced local representations (described
  in Section \ref{sec.chart_boundary}) also carry differential structure
  (via Proposition \ref{lem:embedding given by an extension}). 
  In this way motivations (ii) and (iii) are justified.

  With $\ABBPM{M}$ defined it is possible 
  to define the ``coordinate transformations on the boundary'', 
  see Proposition \ref{prop.coordtransonboundarydef}.
  These coordinate transformations are homeomorphisms on the
  disjoint union used to define the
  topological space. They are extensions of the original coordinate
  transformations on $M$ to the local representations
  of the boundary of $M$.
  While the topological space itself is, in general (see 
  Corollary \ref{cor.Q(M)isamanifoldifenvboundaryman}),
  not a manifold the coordinate transformations on the boundary
  do have a close relationship to genuine coordinate transformations
  in some larger manifold. This relationship is described in
  Section \ref{sec.compatiblecharts}.

  Propositions \ref{prop.subboundaries} 
  and \ref{prop.Q(M)fromphi} and Corollary 
  \ref{cor.Q(M)isamanifoldifenvboundaryman}
  show that if the manifold $M$ can be embedded into some
  larger manifold $M_\phi$ by a map $\phi:M\to M_\phi$
  then the topological boundary of the image of $\phi$, $\partial\phi(M)$,
  can be considered to be a subset of the boundary,
  $\ABBPM{M}$, of $M$.
  Since Penrose's conformal boundary can be presented in such a way
  these results justify motivation (i), see the comments after
  Corollary \ref{cor.Q(M)isamanifoldifenvboundaryman}.

  The boundary of $M$, $\ABBPM{M}$, is similar to the Abstract
  Boundary, see Proposition \ref{prop:ABdescibeofB} and the comments
  after Proposition \ref{prop:ABdescibeofB}. As a consequence it is possible
  to apply Scott and Szekeres' classification of Abstract Boundary points
  to $\ABBPM{M}$. This is done in Section \ref{sec.classification_chart}.
  Note that while $\ABBPM{M}$ and the Abstract Boundary
  are similar, $\ABBPM{M}$ is richer due to the presence
  of the topological space and the expression of the equivalence relation
  in terms of homeomorphisms that extend certain coordinate transformations on 
  $M$. Because the topological space is defined via a quotient map
  the topology on $\ABBPM{M}$ can be studied using the vast number of
  existing results on quotient spaces (in contrast with the topology
  given in
  \cite{citeulike:9599226}). On a more technical note,
  each element of the boundary defined here is an equivalence
  class. These equivalence classes, due to the use of the quotient space
  construction, contain only boundary points not boundary sets, see
  Definition \ref{def.ABBPM}. Because of this certain technicalities
  regarding the Abstract Boundary, e.g.\ \cite[Chapter 9]{Whale2010},
  can be avoided. 

  Section \ref{sec:Examples} presents three examples of the
  construction and classification of representations of
  the boundary;
  the conformal compactification of Minkowski spacetime as it is given in
  \cite[Section 5.1]{Hawking1975Large} (note that this presentation of 
  the boundary differs from Penrose's own publications
  \cite{citeulike:8978534,citeulike:9600515,citeulike:10075981}, in particular
  it
  includes spacelike and
  timelike infinity); Scott and Szekeres' construction of the
  maximal extension 
  of the Curzon solution \cite{Scott1986CurzonI,Scott1986CurzonII}
  and Beyer and Hennig's analysis of the global structure of smooth
  Gowdy symmetric generalized Taub-NUT solutions, \cite{Beyer2011Smooth}.
  The examples were selected to demonstrate application of the material
  in Section \ref{sec:MAIN_chart_boundary} and illustrate how 
  standard global analysis using charts translates into information
  about the boundary.
  It is worthwhile noting that Beyer and Hennig's analysis does not involve
  a closed form of the metric and as a consequence additional
  work would be required before other boundary constructions could be
  applied to their class of spacetimes.

  The charts used in Section \ref{sec:Examples} were constructed to
  provide
  explanitory (and, perhaps, predictive) power regarding the global
  structure of the particular manifolds considered. To do this
  the behaviour of
  geometric and physical quantities on or near local
  representations of $\ABBPM{M}$ were studied,
  \cite{citeulike:8978534,citeulike:9600515,citeulike:10075981,Scott1986CurzonI,Scott1986CurzonII,Beyer2011Smooth}.
  As mentioned above, the construction of $\ABBPM{M}$ provides no
  guidence on the construction of such charts, rather the existing
  techniques of the chart based study of global structure should be
  used. The boundary $\ABBPM{M}$ then provides a way of comparing
  information computed in different charts on ``the boundary
  of the manifold'', hence providing a
  method via which claims of coordiante dependence
  can be counted.

  The paper closes with two Appendices. Appendix \ref{sec.topQ(M)}
  contains proofs of claims made in Section \ref{sec.cicib} and
  Appendix \ref{sec.compcompexprop} gives the proofs of two results
  stated in Section \ref{sec.compatiblecharts}. The results and claims
  that rest on the material in these appendices are clearly highlighted.

  \subsection{Notation and Definitions}\label{sec:notations_and_definitions}
    
    The disjoint union of two sets $U,V$ is denoted by $U\sqcup V$. 

    Given a topological space, $X$, and $U\subset X$ the boundary
    $\overline{U}\setminus \interior\, U$ will be denoted by
    $\partial U$. This is the topological boundary
    of $U$ in $X$. Given $U\subset V\subset X$ the topological
    boundary of $U$ relative to $V$ will be denoted by $\partial_V U$. 
    When $U$ and $V$ are open in $X$, $\partial_VU = V\cap\partial U$.
    Sequences in $X$ are denoted by $(x_i)\subset X$,
    where it is implicitly assumed
    that $i\in\N$. If the sequence $(x_i)$ converges to $x$, I shall
    write $x_i\to x$.
  
    A topological manifold is a locally Euclidean, second-countable, Hausdorff
    topological space. Note that 
    some authors, e.g.\
    \cite{citeulike:12397770}, drop second-countability.
    A manifold, $M$, is a second countable, 
    Hausdorff, topological space 
    equipped with a maximal $C^\infty$ atlas $\Atlas{M}$ of functions from
    $M$ to $\R^n$ where $n$ denotes the dimension of $M$. 
    In order to reduce the proliferation of notation the
    domain of a chart $\alpha\in\Atlas{M}$ will be denoted by $\dom(\alpha)$ 
    and the range
    of $\alpha$ by $\ran(\alpha)$. I will always assume that
    $\dom(\alpha)\subset M$, $\ran(\alpha)\subset \R^n$ and that both sets
    are open.    
    The phrase `manifold with boundary' is used exclusively to denote an 
    $n$-dimensional manifold with $C^\infty$ atlas from $M$ to 
    $\R^{n-1}\times\{x\in\R:x\geq 0\}$. A topological manifold
    with boundary is a second-countable, Hausdorff topological space
    that is locally homeomorphic to  
    $\R^{n-1}\times\{x\in\R:x\geq 0\}$.
    
    An embedding, $\phi:M\to M_\phi$, of $M$ into a manifold $M_\phi$
    of the same dimension as $M$
    is called an envelopment of $M$. 
    When $\hat{M}$ and $M$ are both manifolds,
    $\hat{M}\subset M$ will 
    always imply that $\hat{M}$ is a regular submanifold of $M$.
    Given an embedding $\phi:M\to M_\phi$ the topological boundary
    of $\phi(M)$ in $M_\phi$ is $\partial\phi(M)$, in accordance with the 
    notation above.
    
    A $C^l$ pseudo-Riemannian manifold, $(M,g)$, is a manifold $M$ equipped with
    a symmetric, non-degenerate, $C^l$ two-form $g$. Given a $C^l$ 
    pseudo-Riemannian
    manifold $(M,g)$ and 
    a $C^k$, $k\leq l$, pseudo-Riemannian manifold $(\hat{M},\hat{g})$,
    the pair 
    $(\hat{M},\hat{g})$
    is a $C^k$ extension of $(M,g)$ if $M\subset\hat{M}$ and 
    $\hat{g}|_{TM\times TM}=g$.
    
    A curve, $\gamma:I\to M$, is a $C^0$, piecewise $C^1$, function from 
    an interval $I\subset \R$ to $M$. When expressing relationships between
    subsets or points of $M$ and the image of $\gamma$ the symbol
    $\gamma$ rather than, the more correct, $\gamma(I)$ will be used.
    Hence, $p\in\gamma$
    stands for $p\in\gamma(I)$, $p\in\overline{\gamma}$ stands for
    $p\in\overline{\gamma(I)}$ and $\gamma\cap\dom(\alpha)$ stands for 
    $\gamma(I)\cap\dom(\alpha)$.
    
    The curve $\gamma$ is bounded
    if there exists $b\in\R$ so that for all $x\in I$, $x\leq b$. If $\gamma$
    is not bounded then $\gamma$ is unbounded. 
    A curve $\delta:I_\delta\to M$
    is a sub-curve of $\gamma:I_\gamma\to M$ if $I_\delta\subset I_\gamma$
    and $\gamma|_{I_\delta}=\delta$. A change of parameter on 
    $\gamma:I\to M$ is a monotone
    increasing $C^1$ function $s:J\to I$, where $J\subset\R$ is an
    interval. Two curves $\gamma:I_\gamma\to M$ and
    $\delta:I_\delta\to M$ are related (or obtained) by a change of
    parameter if there exists a change of parameter $s:I_\gamma\to I_\delta$
    so that $\delta\circ s=\gamma$.

\section{A coordinate independent chart induced boundary}
\label{sec:MAIN_chart_boundary}

  This section presents a rigorous formulation of the currently heuristic
  use of charts to study global structure, as outlined in the introduction,
  and uses this to construct a coordinate independent chart induced boundary.
  Section \ref{sec.chart_boundary} formalises what is meant by a chart
  `extending to the edge of a manifold' and defines a class of such
  charts that are suitable for the use in the analysis of quantities defined
  `on the edge'. Section \ref{sec.relbetweenbound} defines an equivalence
  relation that is an extension of coordinate transformations
  between charts to the boundary. In Section \ref{sec.cicib}
  a topological space is constructed using this equivalence relation.
  This topological space is the completion of the manifold by the
  boundary structures induced by the charts that `extend to the edge'.
  The manifold is embedded as a dense open set into this topological space.
  A collection of homeomorphisms are also defined.
  On the image of the manifold these homeomorphisms are the normal
  coordinate transformations. On the topological boundary of the
  manifold in the topological space they can be thought of as
  `coordinate transformations on the boundary'.
  Section \ref{sec.compatiblecharts} studies what conditions are
  needed for the topological space to be a manifold with boundary
  and for the coordinate transformations on the boundary to be
  genuine coordinate transformations.
  Lastly, Section \ref{sec.classification_chart} presents
  a physically motivated classification of the elements of the topological
  boundary of the manifold in the topological space.

  Note that the constructions
  below correspond, in a well defined way, to constructions used in the
  theory of Cauchy spaces. 
  The study of such spaces gives a natural way to
  discuss compactifications and completions, 
  \cite{citeulike:10080487,citeulike:12273874}. 
  Hence it should be no
  surprise that the material of this paper is related to them.
  All of the material of
  this section, excluding the classification, can be rephrased in terms of
  Cauchy spaces. 
  The 
  references \cite{Whale2010,citeulike:8625460}
  contain further discussion of this.

  \subsection{Boundaries induced by charts}\label{sec.chart_boundary}

    This section identifies which points in the topological boundary 
    of the range of a chart can be considered as representing a portion
    of the boundary of the manifold and which charts are suitable for use in
    this way. 
    I show that every suitable chart 
    induces an envelopment of the manifold so that the topological boundary
    given by the envelopment is the same as the boundary given by the 
    chart. Thus each suitable chart also allows for calculations on the
    boundary that require differential structure.
    
    \begin{definition}\label{def:BP}
      Let $M$ be a manifold and $\alpha\in\Atlas{M}$ a chart. 
      An admissible boundary point of $\alpha$ is an 
      element $p\in\partial\ran(\alpha)$
      so that for all sequences $(x_i)\subset \dom(\alpha)$,
      if $\alpha(x_i)\to p$ then $(x_i)$ has no accumulation 
      points in $M$. The set of all admissible
      boundary points of $\alpha$ will be denoted $\BP{\alpha}$.
    \end{definition}  
  
    The elements of $\BP{\alpha}$ can be thought of as 
    those points on the `edge' of 
    $M$ that $\alpha$ extends to. 
    
    \begin{example}\label{ex:bp_ex}
      Let $M=\R^2$ and let
      $\alpha:\R^2\setminus\{(x,0):x\leq 0\}\to\R^2$ be the chart given by 
      \[
        \alpha(x,y)=
          \left(\frac{2}{\pi}\tan^{-1}\left(\sqrt{x^2+y^2}\right),\, \arctan\left({x},y\right)\right),
      \]
      where $\arctan:\R\times\R\to(-\pi,\pi)$ is defined by
      \[
        \arctan(x,y)=\left\{\begin{aligned}
          \tan^{-1}\left(\frac{y}{x}\right) && x>0&\\
          \tan^{-1}\left(\frac{y}{x}\right) + \pi && x<0, y> 0&\\
          \tan^{-1}\left(\frac{y}{x}\right) - \pi && x<0, y<0&\\
          \frac{\pi}{2} && x=0, y>0&\\
          -\frac{\pi}{2} && x =0, y<0&.
        \end{aligned}\right.
      \]
      This chart can be thought of as `compactified polar coordinates'
      on the plane.  
      The range of $\alpha$ is $(0,1)\times\left(-\pi, \pi\right)$.
      The set $\partial\ran(\alpha)$ is the union of the four sets
      \begin{align*}
        S_1 &= \{1\}\times\left[-\pi, \pi\right], &
        S_2 &= \{0\}\times\left[-\pi, \pi\right],\\
        S_3 &= (0,1)\times\left\{-\pi\right\}, &
        S_4 &= (0,1)\times\left\{ \pi\right\}.
      \end{align*}
      A sequence $(x_i)\subset\ran(\alpha)$ with an endpoint in;
      \begin{enumerate}
        \item $S_1$ is such that $(\alpha^{-1}(x_i))$ has no accumulation points,
        \item $S_2$ is such that $(\alpha^{-1}(x_i))$ has $(0,0)$ as an endpoint,
        \item $S_3$ or $S_4$ is such that $(\alpha^{-1}(x_i))$ 
          has an accumulation point in
				  $\{(x,0)\,:\, x<0\}$.
      \end{enumerate}
      Hence $\BP{\alpha}=S_1$. 
    \end{example}
    
    Not every chart with admissible boundary points gives a suitable 
    representation of the boundary. 
  
    \begin{example}\label{ex:improper_chart}
      Let $M=\R^2\setminus\{(x,y)\in\R^2:y\geq\lvert x\rvert\}$. The chart 
      $\alpha:\{(x,y)\in\R^2:y<0\}\to\R^2$ given by $\alpha(x,y)=(x,y)$
      has the point $(0,0)$ as its only admissible boundary point. 
    \end{example}
    
    Charts like $\alpha$ of Example \ref{ex:improper_chart}
    describe a single isolated boundary point. Hence the chart is unable to
    describe topological, let alone differential, structure on the `edge' of
    the manifold. 
    The following definition describes the class of charts that 
    avoid this issue. 
    
    \begin{definition}\label{def:extension}
      Let $M$ be a manifold, $\alpha\in\Atlas{M}$ and $U\subset\R^n$, an 
      open set, so that 
      $\ran(\alpha)\subset U$ and
      $\varnothing\neq\partial_U \ran(\alpha)\subset\BP{\alpha}$.
      Then the pair $(\alpha, U)$ is an extension, $\alpha$ is extendible
      and $U$ is an extension of $\alpha$. Let
      \[
          \EX{M}=\left\{(\alpha,U):\alpha\in\Atlas{M},\, 
          U\text{ is an extension of }\alpha
        \right\}
      \]
      be the set of all extensions.
    \end{definition}

    Propositions \ref{lem:embedding given by an extension} and
    \ref{prop.inverseenvtoext}, below, justify the use of the word extension
    and demonstrate that $\partial_U\ran(\alpha)$ can be thought of as
    a, representation of a portion of the, boundary of $M$.
    Note that if $(\alpha,U)$ is an extension then 
    $\partial_U\ran(\alpha)=\BP{\alpha}\cap U$ as both $\ran(\alpha)$
    and $U$ are open.   

    \begin{example}
      Continuing from Example \ref{ex:bp_ex}:
      let $U = (0,\infty)\times\left(-\pi, \pi\right)$ then
      $\ran(\alpha)\subset U$ and
      $\partial_U\ran(\alpha) = S_1 =\BP{\alpha}$. Thus $(\alpha,U)$ is 
      an extension. 

      Continuing from Example \ref{ex:improper_chart}: every neighbourhood
      of $(0,0)$ contains a point in $\partial \ran(\alpha)$ that is
      not in $\BP{\alpha}$. Therefore $\alpha$ has no extensions.
    \end{example}

    By allowing $U$ to vary, flexibility to remove admissible boundary points
    has been introduced into the definition of an extension. This could be 
    useful if, for some reason, a portion of $\BP{\alpha}$ was considered
    `artificial' while some other portion was considered `natural'.
    Note, however that for any extendible chart there is a 
    maximum extension. 

    \begin{lemma}\label{lem.maxextension}
      Let $\alpha\in\Atlas{M}$ be extendible. Then there exists a
      unique $V\subset\R^n$
      so that $(\alpha, V)\in\EX{M}$ and so that if $(\alpha, U)\in \EX{M}$
      then $U\subset V$.
    \end{lemma}
    \begin{proof}
      Let $\mathcal{C}=\{U\subset\R^n:(\alpha, U)\in\EX{M}\}$.
      Let $V=\bigcup\mathcal{C}$.
      Since $\alpha$ is extendible $\mathcal{C}\neq\varnothing$ and hence
      $V\neq \varnothing$. Since $\ran(\alpha)\subset U$ for all 
      $U\in\mathcal{C}$, $\ran(\alpha)\subset V$. 
      Since $\ran(\alpha)$ and $V$ are open 
      $\partial_V\ran(\alpha)=\bigcup_{U\in\mathcal{C}}\partial_U\ran(\alpha)$.
      Thus $\partial_V\ran(\alpha)\subset \BP{\alpha}$. Lastly since
      $\alpha$ is extendible there exists $U\subset\R^n$ so that
      $(\alpha, U)\in\EX{M}$. Thus $\partial_U\ran(\alpha)\neq\varnothing$
      and so, from above, $\partial_V\ran(\alpha)\neq\varnothing$.
      That is $(\alpha,V)\in \EX{M}$. It is clear, by construction, that
      if $(\alpha, U)\in \EX{M}$ then $U\subset V$. This also
      implies the uniqueness of $V$. Hence we have the
      result.
    \end{proof}

    \begin{definition}\label{def.ABBPM}
      Let $(\alpha,U)\in\EX{M}$ and $V\subset \BP{\alpha}\cap U$. 
      The triple $(\alpha, U, V)$, or sometimes just $V$, is called
      a boundary set. If $V=\{p\}$ then
      the triple $(\alpha,U,\{p\})$, or sometimes just $p$
      is called
      a boundary point.
      The set of all boundary points is 
      \[
        \BPM{M}=\left\{(\alpha,U,\{p\}):(\alpha,U)\in\EX{M},\, 
          p\in\BP{\alpha}\cap U
          \right\}.
      \]  
    \end{definition}

    The chart based approach to studying the `edge' of a manifold chooses
    a chart, $\alpha$, and identifies some subset of $\partial\ran(\alpha)$
    as representing the `edge' of the manifold. In the framework above
    this subset is the set $\partial_U\ran(\alpha)$ associated to
    an extension $(\alpha, U)$ of $\alpha$.
    In this way I have attempted to formalise the chart based approach.

    I now turn to describing the envelopment induced by an extension.
    As mentioned above given an extension $(\alpha,U)\in\EX{M}$  the set 
    $\partial_U\ran(\alpha)=\BP{\alpha}\cap U$ can be thought of as a
    representation of a portion of 
    the boundary of $M$. The following proposition
    makes this idea precise by showing how an extension induces
    an envelopment of the manifold so that the topological boundary
    of the image on the manifold is $\BP{\alpha}\cap U$.
    The range of the envelopment is the topological pasting of $M$
    and $U$ via the diffeomorphism $\alpha$, which is a manifold
    by the definition of an extension.

    \begin{proposition}\label{lem:embedding given by an extension}
      Let $(\alpha,U)\in\EX{M}$ then there exists a manifold 
      $M_{(\alpha,U)}$ and
      an envelopment $\Psi(\alpha,U):M\to M_{(\alpha,U)}$ so that 
      $\partial\Psi(\alpha,U)(M)=\BP{\alpha}\cap U$.
    \end{proposition}
    \begin{proof}
      Let $(\alpha,U)\in\EX{M}$ and
      let $\alpha_U:\dom(\alpha)\cup (U\setminus \ran(\alpha))\to U$ 
      be defined by
      \[
        \alpha_U(x)=\left\{\begin{aligned}
          \alpha(x) && x\in \dom(\alpha)\hphantom{.}\\
          x && x\in U\setminus \ran(\alpha).
        \end{aligned}\right.
      \]
      Give the set $M_{(\alpha,U)}=
      M\cup\{U\setminus \ran(\alpha)\}$
      the topology
      \begin{multline*}
        \mathcal{T}=\left\{U\subset M:U\text{ is open in }M\right\}\cup\\
          \left\{\alpha_U^{-1}(V)\subset \dom(\alpha)\cup(U\setminus \ran(\alpha)):
          V \text{ is open in }U\right\}.
      \end{multline*}
      The topological space
      $M_{(\alpha,U)}$
      equipped with the atlas generated by the functions
      $\Atlas{M}\cup\{\alpha_U\}$
      is a manifold. The map 
      $\Psi\left(\alpha,U\right):M\to M_{(\alpha,U)}$ given by 
      $\Psi\left(\alpha,U\right)(x)=x$
      is the required envelopment.
      The proofs of these claims are straightforward,
      and follow directly from the definition of $\mathcal{T}$,
      the definition of
      $\EX{M}$ and as $\alpha$ is a diffeomorphism. 
      The details are left to the reader.
      
      I now show that 
      $\partial\Psi\left(\alpha,U\right)(M)=\BP{\alpha}\cap U$.
      Let $x\in\partial\Psi\left(\alpha,U\right)(M)$. 
      By construction this implies that 
      $x\in U\setminus\ran(\alpha)$. If
      $x$ is in the interior of $U\setminus\ran(\alpha)$ then
      there exists $V$ an open subset of $U$ so that
      $V\subset U\setminus\ran(\alpha)$. The definition of
      $\mathcal{T}$ implies that 
      $x\not\in \partial\Psi\left(\alpha,U\right)(M)$. Thus $x$
      is not in the interior of $U\setminus\ran(\alpha)$. Since $U$
      and $\ran(\alpha)$ are open this implies that 
      $x\in\partial_U\ran(\alpha)$. Thus there exists
      $(x_i)\subset\dom(\alpha)$ so that $(\alpha(x_i))\to x$.
      The definition of $\mathcal{T}$ implies that
      $(\Psi\left(\alpha,U\right)(x_i))\to x$. Since $M_{(\alpha,U)}$
      is Hausdorff $(\Psi\left(\alpha,U\right)(x_i))$ has no limit points
      in $\Psi\left(\alpha,U\right)(M)$. 
      As $\Psi(\alpha,U)$ is a homeomorphism this implies that
       $(x_i)$ has no limit points 
      in $M$. Therefore $x\in\BP{\alpha}$. 
      From above, $x\in\partial_U\ran(\alpha)\subset U$.
      That is $\partial\Psi\left(\alpha,U\right)(M)\subset\BP{\alpha}\cap U$.
            
      Suppose that $x\in\BP{\alpha}\cap U$. By definition there exists
      $(x_i)\subset\dom(\alpha)$ so that $(\alpha(x_i))\to x$ and
      $(x_i)$ has no accumulation points in $M$. 
      The definition of $\mathcal{T}$
      implies that $(\Psi(\alpha,U)(x_i))\to x$. Hence
      $x\in\overline{\Psi(\alpha,U)(M)}$. If $x\in\Psi(\alpha,U)(M)$
      then, as $\Psi(\alpha,U)$ is a homeomorphism, $(x_i)$ would have 
      $\Psi(\alpha,U)^{-1}(x)$ as an accumulation point. This is a
      contradiction, however. Hence $x\in\partial\Psi(\alpha,U)(M)$
      and therefore 
      $\BP{\alpha}\cap U\subset\partial\Psi\left(\alpha,U\right)(M)$.
      
      Thus $\partial\Psi\left(\alpha,U\right)(M)=\BP{\alpha}\cap U$ as 
      required.
    \end{proof}
    There is a converse to the above Proposition.

    \begin{proposition}\label{prop.inverseenvtoext}
      Let $\phi:M\to M_\phi$ be an envelopment
      and let $\alpha\in\Atlas{M_\phi}$.
      If $\dom(\alpha)\cap\partial\phi(M)\neq\varnothing$
      then $(\alpha\circ\phi, \ran(\alpha))\in\EX{M}$, 
      where $\rest{\alpha}{\phi}$ denotes the map
      $
        \rest{\alpha}{\phi}=\alpha\circ
          \phi|_{\phi^{-1}(\dom(\alpha)\cap\phi(M))}.
      $
    \end{proposition}
    \begin{proof}
      It is clear that $\ran(\alpha\circ\phi)\subset\ran(\alpha)$.
      Since $\dom(\alpha)\cap\partial\phi(M)\neq\varnothing$
      we know that 
      $\partial_{\ran(\alpha)}\ran(\alpha\circ\phi)\neq\varnothing$.
      Let $x\in\partial_{\ran(\alpha)}\ran(\alpha\circ\phi)$.
      Let 
      $(x_i)\subset\dom(\alpha\circ\phi)$ be a sequence
      so that $\alpha\circ\phi(x_i)\to x$. 
      Since $\alpha$ is a diffeomorphism
      $\phi(x_i)\to \alpha^{-1}(x)$. By assumption 
      $x\in\partial_{\ran(\alpha)}\ran(\alpha\circ\phi)$
      which implies that $\alpha^{-1}(x)\in\partial\phi(M)$.
      Since $\phi$ is a diffeomorphism and $M_\phi$
      is Hausdorff this implies that the
      sequence $(x_i)$ has no limit points in $M$.
      Therefore
      $x\in\BP{\alpha\circ\phi}$. Hence
      $\partial_{\ran(\alpha)}\ran(\alpha\circ\phi)\subset\BP{\alpha\circ\phi}$
      as required.
    \end{proof}

    These two results show that, for any extension
    $(\alpha, U)$ it makes sense to think of $\BP{\alpha}\cap U$
    as the representation of a portion of the boundary of 
    $M$ given by the chart $\alpha$.
    This thought can only be taken so far, however, as there is no guarantee that
    $\overline{\Psi(\alpha,U)(M)}$ is a manifold with boundary.
     
    \begin{example}\label{ex:sincurve}
      Let $M=\left\{(x,y):y<\sin\left(\frac{1}{x}\right), x>0\right\}$
      and define the chart $\alpha:M\to\R^2$ by $\alpha(x,y)=(x,y)$.
      Let $U=\R^2$ then $(\alpha,U)\in\EX{M}$ and $M_{(\alpha,U)}=\R^2$.
      The set $\overline{\Psi(\alpha,M)(M)}$ is
      \[
        M\cup\{(0,y):y<-1\}\cup
          \overline{\left\{\left(x,\,\sin\left(\frac{1}{x}\right)\right):
          x>0\right\}}.
      \]
      The last set in the union contains points, $\{(0,y)\in\R^2: -1<y\leq 1\}$, 
      that are not the endpoint of any curve
      in $M$.
      As a consequence $\overline{\Psi(\alpha,U)(M)}$ is not a manifold with
      boundary. The space $\overline{\Psi(\alpha,U)(M)}$ must be 
      considered as a topological space which has an open dense subset which
      is a manifold. 
    \end{example}

  \subsection{Relationships between boundaries given by extensions}
  \label{sec.relbetweenbound}

    The boundary $\BP{\alpha}\cap U$ of an extension $(\alpha,U)$ is, 
    naturally,
    dependent on $(\alpha,U)$ and thus is not chart 
    independent. 
    Of course the same is true of $\ran(\alpha)$. This is a chart
    dependent way of viewing the manifold. In the case of $\ran(\alpha)$
    the solution to chart dependence
    is a part of the basic apparatus of differential geometry:
    quantities defined on $\ran(\alpha)$ are well defined on $M$ if
    they transform the `right way' under changes of coordinates.

    In the context of this paper, where I am concerned with a construction
    of a topological space (the boundary of $M$), transforming the right way
    means that $x\in\ran(\alpha)$ represents the same point in the manifold
    as $y\in\ran(\beta)$ if and only if $\beta\circ\alpha^{-1}(x)=y$.
    Let $(\alpha, U)$ and $(\beta, X)$ be extensions.
    If the function
    $\beta\circ\alpha^{-1}:\alpha(\dom(\alpha)\cap\dom(\beta))\to
    \beta(\dom(\alpha)\cap\dom(\beta)$ could be extended to some
    portion of $\partial_U\ran(\alpha)$ and $\partial_X\ran(\beta)$
    then this extended function would give a natural way to
    think about coordinate invariance on the boundary of the manifold.
    The example below shows that it is not always possible 
    to extend $\beta\circ\alpha^{-1}$ to all of
    $\partial _U\ran(\alpha)$ as a function. 
    It is, however, always possible to extend it as a relation.
    
    Each point in $\ran(\alpha)\cup\partial_U\ran(\alpha)$ can be
    uniquely equated with the set of sequences in $\ran(\alpha)$ that
    limit to the point. Since each sequence is in $\ran(\alpha)$
    we can apply $\beta\circ\alpha^{-1}$ to it. If the image of
    the sequence has a limit point, then the relation extending
    $\beta\circ\alpha^{-1}$ will relate the point in 
    $\ran(\alpha)\cup\partial_U\ran(\alpha)$ and the limit point.
    This relation can be thought of as  
    the `coordinate transform on the edge of the
    manifold' between $\alpha$ and $\beta$, though in Proposition 
    \ref{prop.coordtransonboundarydef} I define functions which are
    better suited to this terminology.

    Because of the use of a relation and not a function it is
    possible for a boundary point of some extension to correspond to a boundary
    set of another. This is a down side to the construction, which
    can be addressed through clever choices of charts, see Section 
    \ref{sec.compatiblecharts}.
    The example below, which illustrates this behaviour,
    is an expansion of a brief discussion 
    given in Geroch's 1968 paper
    \cite{citeulike:4251705}.

    After the example the remainder of this section will define
    the relation, alluded to above, and discuss when a boundary point/set
    should correspond to a boundary point/set in some other extension.

    \begin{example}[{\cite{Geroch1968Local}}]\label{expointtoBviceversa}
      Let $h:[0,\pi)\to [0,2)$ be the smooth function defined by
      \[
        h(\theta)=\left\{
          \begin{aligned}
            & 0 && \theta = 0\\
            &\frac{1}{1+\exp\left(\frac{\pi}{4}\left(\frac{1}{\theta}-
              \frac{4}{\pi-4\theta}\right)\right)} && 
              \theta\in \left(0,\frac{\pi}{4}\right)\\
            &1 && \theta\in\left[\frac{\pi}{4},\frac{3\pi}{4}\right]\\
            &1+\frac{1}{1+\exp\left(\frac{\pi}{4}\left(
              \frac{1}{\theta -\frac{3\pi}{4}}
              -\frac{4}{\pi-4\left(\theta -\frac{3\pi}{4}\right)}
              \right)\right)} && 
              \theta\in\left(\frac{3\pi}{4}, \pi\right)\\
          \end{aligned}\right.
      \]
      and let $f:[0,1)\times [0,\pi)\to[0,1)\times[0,\pi)$ be defined by
      $f(r, \theta)=(r,\theta(1-r) + r\frac{\pi}{2}h(\theta))$.
      The function $f$ can be thought of as a homotopy between
      the constant function $\theta\mapsto\theta$ and
      a function which maps $[0,\pi)$ onto itself while
      contracting
      $\left[\frac{\pi}{4},\frac{3\pi}{4}\right]$ 
      to the point $\frac{\pi}{2}$.
      The function $f$ is a smooth bijection hence it's 
      inverse, $f^{-1}$, is well defined.
      
      Let $M=\{(x,y):x^2+y^2<1\}$ and define the following charts
      $\alpha:M\to\R^2$ and $\beta:M\to\R^2$ given by, in polar coordinates,
      $\alpha(r,\theta)=(r,\theta)$ 
      and
      \[
        \beta(r,\theta)=\left\{
          \begin{aligned}
          &f(r,\theta) && \theta\in [0,\pi)\\
          &f^{-1}(r, \theta -\pi)+(0,\pi) && \theta\in [\pi,2\pi).
        \end{aligned}\right.
      \]
      Then the transition function $\phi=\beta\circ\alpha^{-1}$ is 
      given by $\phi=\beta$. The set $\R^2$ extends both
      $\alpha$ and $\beta$.

      For each $i\in\N$, let 
      \begin{align*}
        u_i&=\left(\frac{i}{i+1},\, 
          \frac{5\pi}{4} + \frac{\pi}{4}\frac{i}{i+1}\right),
          &          
        v_i&=\left(\frac{i}{i+1},\, 
          \frac{7\pi}{4} - \frac{\pi}{4}\frac{i}{i+1}\right).
      \end{align*}
      Note that these sequences have the following limits
      \begin{align*}
        u_i &\to \left(1,\, \frac{3\pi}{2}\right),
        & v_i &\to \left(1,\, \frac{3\pi}{2}\right).
      \end{align*}
      The images of these sequences under $\phi$ are
      \begin{align*}
        \phi(u_i)&=\left(\frac{i}{i+1},\, \frac{5\pi}{4}\right),
        & \phi(v_i)&=\left(\frac{i}{i+1},\, \frac{7\pi}{4}\right).
      \end{align*}
      The limits of these images are 
      \begin{align*}
        \phi(u_i) &\to \left(1,\, \frac{5\pi}{4}\right),
        & \phi(v_i) &\to \left(1,\, \frac{7\pi}{4}\right).
      \end{align*}
      Hence while $u_i,v_i$ identify the same boundary point
      in $(\alpha, \R^2)$ they identify different boundary points
      in $(\beta, \R^2)$. This implies that there is no continuous
      extension of $\phi$ to a neighbourhood of $(1,\frac{3\pi}{2})$.
      In accordence with the discussion
      above, the boundary point
      $\left(1,\frac{3\pi}{2}\right)$ of $(\alpha, \R^2)$ is considered
      to be
      represented by a set of boundary points with respect to
      $(\beta,\R^2)$ which contains the points $(1,\frac{5\pi}{4})$
      and $(1,\frac{7\pi}{4})$. 
      It is possible to explicitly identify
      this set.

      Let $X =
      \left\{(1,\theta):\theta\in
      \left[\frac{5\pi}{4},\frac{7\pi}{4}\right]\right\}$.
      Similar calculations, as above,
      show that every sequence $(x_i)\subset\ran(\beta)$
      that converges
      to a point in $X$ is such that
      $\phi^{-1}(x_i)\to\left(1,\frac{3\pi}{2}\right)$. As well
      as that every sequence, $(x_i)\subset\ran(\alpha)$
      that converges to $\left(1,\frac{3\pi}{2}\right)$ is such that
      the accumulation points of 
      the sequence $(\phi(x_i))$ lie in $X$.
      Hence, considering boundary points to be sets of sequences 
      we should consider $\left(1,\frac{3\pi}{2}\right)$
      and $X$ as different representations of the same boundary point
      with respect to the extensions $(\alpha, \R^2)$ and $(\beta, \R^2)$.

      Now consider an element $x\in X$. From the discussion above
      it is clear that if $(x_i)\subset\ran(\beta)$ is
      such that $x_i\to x$ then 
      $\phi^{-1}(x_i)\to \left(1,\frac{3\pi}{2}\right)$.
      Thus the set of sequences which identify $x$, under $\phi^{-1}$, can 
      be viewed as a subset
      of the set of sequences which identify $\left(1,\frac{3\pi}{2}\right)$.
      
      The same arguments can be used to show that the
      boundary point $\left(1,\frac{\pi}{2}\right)$ of $(\beta, \R^2)$
      and the set of boundary points
      $\left\{(1,\theta):\theta\in
      \left[\frac{\pi}{4},\frac{3\pi}{4}\right]\right\}$ of
      $(\alpha,\R^2)$ should be considered as 
      different representations of the same portion of the
      boundary point with respect to the extensions
      $(\beta,\R^2)$ and $(\alpha,\R^2)$.
    \end{example}

    The definition of the relation, mentioned above, is as follows.

    \begin{definition}\label{def:chartcoversing}
      The boundary set $(\alpha,U,V)$ covers the boundary set 
      $(\beta,X,Y)$,
      $(\alpha,U,V)\covers(\beta,X,Y)$, if for all 
      $(y_i)\subset\dom(\beta)$
      so that $(\beta(y_i))$ has an accumulation point in $Y$,
      there exists a subsequence $(v_i)$ of
      $(y_i)$ so that
      $(v_i)\subset\dom(\alpha)$
      and $(\alpha(v_i))$ has an accumulation point in $V$.
    \end{definition}

    Hence if $(\alpha, U, V)\covers (\beta, X, Y)$ then the
    map $\alpha\circ\beta^{-1}$ can be considered as
    mapping $Y$ into $V$.

    \begin{definition}\label{def.equiv}
      The boundary sets $(\alpha,U,V)$ and $(\beta, X, Y)$ are equivalent,
      $(\alpha,U,V)\equiv(\beta,X,Y)$
      if and only if 
      $(\alpha,U,V)\covers(\beta,X,Y)$ and $(\beta,X,Y)\covers(\alpha,U,V)$.
      The equivalence class of $(\alpha,U,V)$ is denoted $[(\alpha,U,V)]$.
      
      The pre-order $\covers$ induces a partial order, which in an abuse of 
      notation is also denoted by $\covers$, on the set of all 
      equivalence classes: the equivalence class of boundary sets
      $[(\alpha,U,V)]$ covers the equivalence class of boundary sets 
      $[(\beta,X,Y)]$, $[(\alpha,U,V)]\covers[(\beta,X,Y)]$ if and only if
      $(\alpha, U,V)\covers(\beta,X,Y)$.
    \end{definition}

    \begin{example}
      Continuing from Example \ref{expointtoBviceversa}: the
      boundary point 
      $
        \left(\alpha, \R^2, \left\{\left(1,\frac{3\pi}{2}\right)\right\}\right)
      $
      is equivalent to the boundary set 
      $(\beta, \R^2, X)$ and for any $x\in X$ it covers the
      boundary point $(\beta, \R^2, \{x\})$. Similarly,
      the boundary point $\left(\beta, \R^2,
      \left\{\left(1,\frac{\pi}{2}\right)\right\}\right)$ is equivalent to the
      boundary set $\left(\alpha,
      \R^2,\left
      \{(1,\theta):\theta\in\left[\frac{\pi}{4},\frac{3\pi}{4}\right]\right\}
      \right)$
      and for  point $y$ in 
      $\{(1,\theta):\theta\in\left[\frac{\pi}{4},\frac{3\pi}{4}\right]\}$
      it covers $(\alpha, \R^2, \{y\})$.
    \end{example}

    If $(\alpha, U)\in\EX{M}$ and $x\in\BP{\alpha}\cap U$, then
    $(\alpha, U, \{x\})$ is a representation of $[(\alpha, U, \{x\})]$.
    The equivalence class $[(\alpha, U, \{x\})]$ can be thought of
    as the coordinate independent object that $\alpha$ provides a 
    representation of. 

  \subsection{The chart induced coordinate invariant boundary}\label{sec.cicib}

    I am now in a position to construct the space which each extension
    $(\alpha, U)$ gives a representation of.

    \begin{definition}\label{def.quotientmap}
      Let $Q\subset\EX{M}$ be a set of extensions and
      let $P_Q=\{(\alpha,\ran(\alpha)):\alpha\in\Atlas{M}\}$.
      While no element of $P_Q$ is an extension
      (as $\varnothing = \partial_{\ran(\alpha)}\ran(\alpha)$), 
      a pair $(\alpha,\ran(\alpha))\in P_Q$
      can be viewed as the trivial extension of $\alpha$.
      Let $S_Q= P_Q\cup Q$.
      For each pair $(\alpha, U)\in S_Q$ let
      $N_{(\alpha, U)}=\ran(\alpha)\cup\partial_U\ran(\alpha)$.
      Each $N_{(\alpha, U)}\subset \R^{n}$ and thus carries the relative
      topology induced by the topology on $\R^n$.
      Let
      \[
        N_Q=\bigsqcup_{(\alpha, U)\in S_Q}N_{(\alpha, U)}.
      \]
      be the disjoint of all $N_{(\alpha, U)}$.
      Define an equivalence relation on $N_Q$ by
      $x\in N_{(\alpha, U)}$ is equivalent to $y\in N_{(\beta, X)}$
      if and only if either $x\in\ran(\alpha),y\in\ran(\beta)$ and
      $\beta\circ\alpha^{-1}(x)=y$ or
      $x\in\partial_{U}\ran(\alpha), y\in\partial_X\ran(\beta)$
      and
      $(\alpha, U, \{x\})\equiv (\beta, X, \{y\})$.
      In an abuse of notion,
      let $Q(M)$ be the image 
      of $N_Q$ under this equivalence relation. 
      There exists a quotient map $q:N_Q\to Q(M)$ that takes each element
      $x\in N_Q$ to its equivalence class $q(x)=[x]$.
      The set $N_Q$ has a natural topology induced by the topologies
      on each of its components.
      Equip $Q(M)$ with the quotient topology induced by the topology
      on $N_Q$ and the map $q$. Hence $Q(M)$ is a topological space.
      The set $Q(M)$ will be called the completion of $M$ with respect to
      $Q$.
      There exists a map $\imath_Q:M\to Q(M)$ given by $\imath_Q(x)=[\alpha(x)]$
      for any $\alpha\in \Atlas{M}$ so that $x\in\dom(\alpha)$.
      The map $\imath_Q$ is well defined by definition of $Q(M)$.
    \end{definition}

    The properties of $q, \imath_Q$ and $Q(M)$ are described in Appendix
    \ref{sec.topQ(M)}.
    For the convenience of the reader I summarise them here.
    The map $q$ is open and continuous and
    the restriction $q|_{N_{(\alpha,U)}}$,
    for all $(\alpha, U)\in S_Q$, is a homeomorphism. 
    The map $\imath_Q$ is
    continuous and a homeomorphism onto its image.
    The space $Q(M)$ is a $T_1$, separable, first countable, 
    locally metrizable
    topological space.  Lastly $\imath_Q(M)$ is an open dense
    subset of $Q(M)$.

    The boundary $\partial\imath_Q(M)$ is what you might expect it
    to be from the discussion in Section \ref{sec.chart_boundary}.

    \begin{proposition}
      Let $Q\subset\EX{M}$ and let $\imath_Q$ be the function
      defined above. Then
      \[
        \partial\imath_Q(M)=\bigcup_{(\alpha, U)\in Q} 
          q(\partial_U\ran(\alpha)) = \bigcup_{(\alpha, U)}q(\BP{\alpha}\cap
          U).
      \]
    \end{proposition}
    \begin{proof}
      The second equality follows by definition of $\EX{M}$, 
      Definition \ref{def:extension}.

      Suppose that there exists $(\alpha, U)\in Q$
      and $x\in\ran(\alpha)$ so that $q(x)\in\partial\imath_Q(M)$.
      By assumption $\alpha^{-1}(x)\in M$ and
      $\imath_Q(\alpha^{-1}(x))=q(x)$. Since $\imath_Q$ is a homeomorphism,
      Proposition \ref{prop.imathconstruct},
      and $M$ is open
      this implies that $q(x)\not\in\partial\imath_Q(M)$, a contradiction.

      Suppose that there exists $(\alpha, U)\in Q$
      and $x\in \partial_U\ran(\alpha)$ so that
      $q(x)\not\in\partial\imath_Q(M)$.
      By Proposition \ref{prop.imathconstruct} 
      $Q(M)=\imath_Q(M)\cup\partial\imath_Q(M)$
      thus $q(x)\in \imath_Q(M)$.
      As $\imath_Q$ is a homeomorphism onto its image this implies
      that there exists $y\in M$ so that $q(y)=q(x)$.
      This implies that $x\not\in\BP{\alpha}$. This is a contradiction
      since $(\alpha, U)\in \EX{M}$.
    \end{proof}

    Since $Q(M)$ is intended to be a completion of $M$ 
    a full investigation of
    the compactness properties of $Q(M)$ would be interesting. 
    Unfortunately the, in general, non-Hausdorff behaviour of points in
    $Q(M)$ make such a study a lengthy and technical affair, thus
    placing it beyond the scope of this paper.
    One can prove, however, that $Q(M)$ is a completion of
    $M$, in the sense of \cite[Definition 1.4.2]{citeulike:10080487},
    once $M$ is equipped with the Cauchy structure
    induced by all sequences in $M$ whose images, 
    under some chart in $Q$, are
    convergent.
    For the purposes of this paper, and most
    applications, Proposition \ref{prop.q(M)conv} is
    sufficient.

    The space $Q(M)$ does not, in general, carry a differential structure
    and, in general, is not 
    locally homeomorphic to $\R^{n-1}\times\{x\geq 0:x\in\R\}$.
    However, the functions
    $q|_{N_{(\alpha, U)}}:N_{(\alpha, U)}\to q(N_{(\alpha, U)})$ 
    are homeomorphisms onto their images in $Q(M)$.
    By Proposition 
    \ref{lem:embedding given by an extension} the set
    $\alpha_U^{-1}(N_{(\alpha, U)})$ is a subset of
    $M_{(\alpha, U)}$. Thus $Q(M)$ is locally homeomorphic to
    suitable subspaces of manifolds. 
    What prevents $Q(M)$ from being a, in general, non-second countable,
    non-Hausdorff manifold is that the maps
    $q|_{N_{(\alpha, U)}}:N_{(\alpha, U)}\to q(N_{(\alpha, U)})$
    are not charts because, in general, $N_{(\alpha,U)}$
    is not an open subset of
    $\R^{n-1}\times\{x\geq 0:x\in\R\}$.
    Never-the-less each $N_{(\alpha, U)}\subset U$ and
    so as $U=\ran(\alpha_U)$, each $N_{(\alpha, U)}$
    can be considered as a chart-like structure
    on $Q(M)$ that carries a differential
    structure induced by $\alpha$.
    In Section \ref{sec.compatiblecharts} conditions under which
    $Q(M)$ is a manifold with boundary are presented.
    In these circumstances $q|_{N_{(\alpha, U)}}$ is a diffeomorphism
    and hence the atlas of $Q(M)$ is generated by the
    maps $q|_{N_{(\alpha, U)}}$.  

    Since Definition \ref{def.quotientmap} and the
    results of Appendix \ref{sec.topQ(M)}
    only require a subset of extensions they can be used to, for example,
    define the $C^1$ boundary of $M$, by requiring that
    each selected extension, $(\alpha, U)$ to be such that
    $\partial_U\ran(\alpha)$ is $C^1$ in $\R^n$.

    Since $Q(M)$ is composed of equivalence classes under $\equiv$
    the points in $Q(M)$ are manifestly coordinate independent, in the 
    sense that two boundary points 
    $(\alpha, U, \{x\})$ and $(\beta, V, \{y\})$
    are representations of the same third, coordinate independent, point
    $[(\nu, W, \{z\})]\in Q(M)$ if and only if
    $(\alpha, U, \{x\})\equiv (\beta, V, \{y\})$. This is a generalisation
    of the coordinate independence described in the beginning of Section
    \ref{sec.relbetweenbound}.
    The condition $(\alpha, U, \{x\})\equiv (\beta, V, \{y\})$ can be rewritten
    as $q(x)=q(y)$, or perhaps even more evocatively as
    $y=\left(q|_{N_{(\beta, V)}}\right)^{-1}\circ q|_{N_{(\alpha, U)}}(x)$.
    The function $\left(q|_{N_{(\beta, V)}}\right)^{-1}\circ q|_{N_{(\alpha, U)}}$
    restricted to $\alpha(\dom(\alpha)\cap\dom(\beta))$ is $\beta\circ\alpha^{-1}$.
    Hence the function
    $\left(q|_{N_{(\beta, V)}}\right)^{-1}\circ q|_{N_{(\alpha, U)}}$ 
    is an extension
    of the coordinate transformation $\beta\circ\alpha^{-1}$ to the
    boundary of $M$.
    Functions of this form are the promised ``coordinate transformations
    on the boundary''.

    \begin{proposition}\label{prop.coordtransonboundarydef}
      Let $Q\subset\EX{M}$ be a set of extensions.
      For every pair $(\alpha,U),(\beta, V)\in S_Q$, there exists a
      homeomorphism
      \begin{align*}
        \overline{\beta\circ\alpha^{-1}}:
        q^{-1}\left(q(N_{(\alpha, U)})\cap q(N_{(\beta, V)})\right)&
        \cap 
        N_{(\alpha, U)}\\
        &\hspace{-1.5cm}\to
        q^{-1}\left(q(N_{(\alpha, U)})\cap q(N_{(\beta, V)})\right)
        \cap 
        N_{(\beta, U)}
      \end{align*}
      given by
      \[
        \overline{\beta\circ\alpha^{-1}}(x)
        =
        \left(q|_{N_{(\beta, V)}}\right)^{-1}\circ q|_{N_{(\alpha, U)}}(x),
      \]
      so that
      $\overline{\beta\circ\alpha^{-1}}|_{\alpha(\dom(\alpha)\cap\dom(\beta))}=
        \beta\circ\alpha^{-1}$ and
      $\overline{\beta\circ\alpha^{-1}}(x)=y$ if and only
      if $q(x)=q(y)$ and $x\in N_{(\alpha, U)}$, 
      $y\in N_{(\beta, V)}$, i.e.\ either $\beta\circ\alpha^{-1}(x)=y$
      or $(\alpha, U, \{x\})\equiv(\beta, V, \{y\})$.
    \end{proposition}
    \begin{proof}
      By Proposition \ref{prop.qrestishomeo}, for all $(\alpha, U)\in S_Q$,
      $q|_{N_{(\alpha, U)}}$ is a homeomorphism.
      Hence $\overline{\beta\circ\alpha^{-1}}$ is a homeomorphism.
      Suppose that $\overline{\beta\circ\alpha^{-1}}(x)=y$.
      Then, by definition and as $x\in N_{(\alpha, U)}$, $y\in N_{(\beta, V)}$
      $
        y=\left(q|_{N_{(\beta, V)}}\right)^{-1}\circ q|_{N_{(\alpha, U)}}(x)
      $
      implies that $q(y)=q(x)$ as required. Similarly, if
      $x\in N_{(\alpha, U)}$, $y\in N_{(\beta, V)}$ and
      $q(x)=q(y)$ then
      $
        y=\left(q|_{N_{(\beta, V)}}\right)^{-1}\circ q|_{N_{(\alpha, U)}}(x)
      $
      by definition of $q$.

      I now show that
      $\overline{\beta\circ\alpha^{-1}}|_{\alpha(\dom(\alpha)\cap\dom(\beta))}=
        \beta\circ\alpha^{-1}$.  By definition of $q$, if 
      $x\in\alpha(\dom(\alpha)\cap \dom(\beta))$ then there exists
      $y\in\beta(\dom(\alpha)\cap\dom(\beta))$ so that
      $\beta\circ\alpha^{-1}(x)=y$. This implies that
      $q(x)=q(y)$.
      Since $x\in N_{(\alpha, U)}$ and $y\in N_{(\beta, V)}$
      this implies that $q(x)\in q(N_{(\alpha, U)})\cap q(N_{(\beta, V)})$
      and therefore that
      $x$ is in the domain of $\overline{\beta\circ\alpha^{-1}}$.
      That is,
      $\alpha(\dom(\alpha)\cap\dom(\beta))\subset
      q^{-1}(q(N_{(\alpha, U)})\cap q(N_{(\beta, V)}))$.
      Thus $\overline{\beta\circ\alpha^{-1}}$ is defined on
      $\alpha(\dom(\alpha)\cap\dom(\beta))$. The definition
      of $q$ now implies that 
      $\overline{\beta\circ\alpha^{-1}}|_{\alpha(\dom(\alpha)\cap\dom(\beta))}=
        \beta\circ\alpha^{-1}$, as required.
    \end{proof}

    The homeomorphism
    $\overline{\beta\circ\alpha^{-1}}$, due to its definition,
    can also be thought of
    as a map which describes the relation $\equiv$ between
    $(\alpha, U)$ and $(\beta, V)$ restricted to boundary points and
    which respects the topologies given by $U$ and $V$.

    As already mentioned
    Proposition \ref{prop.coordtransonboundarydef} allows the set
    $Q$ to be free. 
    The next result describes how the boundaries induced by different
    choices of sets of extensions are related. Note that this is not
    the most general form of the result possible, but it is sufficient 
    for our needs.

    \begin{proposition}\label{prop.subboundaries}
      Let $A,B\subset \EX{M}$ so that for all $(\alpha, U)\in A$
      there exists $(\alpha, V)\in B$ with $U\subset V$.
      Then there is an injective continuous open
      function $f:A(M)\to B(M)$ defined by
      $f([x])=[x]$.
    \end{proposition}
    \begin{proof}
      Let $(\alpha, U)\in A$ and
      let $i:N_{(\alpha, U)}\to N_{(\alpha, V)}$ be the induced
      inclusion.
      The function $f$ can alternatively be given by
      $f|_{q(N_{(\alpha,U)})}([x])=q|_{N_{(\alpha, V)}}\circ i\circ
      q|_{N_{(\alpha, U)}}^{-1}([x])$.
      The claimed properties are now easy to check.
    \end{proof}

    We now give the coordinate independent chart induced boundary, which
    is the largest boundary with respect to the ordering of sub-boundaries
    given by Proposition \ref{prop.subboundaries}.

    \begin{definition}[The coordinate independent chart induced boundary]
    \label{def:chart_boundary}
      For each $\alpha\in \Atlas{M}$ that is extendible
      let $U_\alpha$ be the maximal extension given by
      Lemma \ref{lem.maxextension}.
      Let $Q_{\mathcal{A}}=\{(\alpha,U_\alpha):\alpha\in\Atlas{M}
        \text{ is extendible}\}$.
      Let $Q_{\mathcal{A}}(M)$ and 
      $\imath_{Q_{\mathcal{A}}}:M\to Q_{\mathcal{A}}(M)$ be as given in Definition 
      \ref{def.quotientmap}.
      The boundary of $M$ is $\partial\imath_{Q_{\mathcal{A}}}(M)$.
      To emphasise the existence of the boundary independently from
      $\imath_{Q_{\mathcal{A}}}$ 
      I denote $\partial\imath_{Q_{\mathcal{A}}}(M)$ by $\ABBPM{M}$.
      The set $Q_{\mathcal{A}}(M)$
      will be called
      the completion of $M$.
    \end{definition}

    \begin{definition}
      In an abuse of notation, the elements 
      of $\ABBPM{M}$ are called boundary points. The distinction between
      a boundary point $[(\alpha,W,\{p\})]\in\ABBPM{M}$ and 
      a boundary point $(\alpha,W,\{p\})\in\BPM{M}$ will be clear from context.
    \end{definition}

    Due to the structure of $Q_{\mathcal{A}}(M)$, quantities defined using the 
    standard chart based approach to studying the boundary of a manifold
    will be well defined on the completion of $M$
    if they respect the equivalence relation
    used to define $Q_{\mathcal{A}}(M)$. 
    On $\imath_{Q_{\mathcal{A}}}(M)$ this is nothing other than the
    usual invariance under changes of coordinates. On $\ABBPM{M}$ this
    can be interpreted as invariance under coordinate transformations
    on the boundary, as given by Proposition 
    \ref{prop.coordtransonboundarydef}.

    The space $Q_{\mathcal{A}}(M)$ is a completion of $M$ in the following sense.

    \begin{proposition}\label{prop.q(M)conv}
      Let $Q_{\mathcal{A}}(M)$ be the completion of $M$, 
      as given in Definition \ref{def:chart_boundary}.
      If $(x_i)\subset M$ is a sequence then 
      $(\imath_{Q_{\mathcal{A}}}(x_i))$ has at least one limit point.
    \end{proposition}
    \begin{proof}
      If $x\in M$ is a limit point of $(x_i)$ then as
      $\imath_{Q_{\mathcal{A}}}$ is continuous, and 
      $Q_{\mathcal{A}}(M)$ is first countable,
      $\imath_{Q_{\mathcal{A}}}(x)$ is a limit point of
      $\imath_{Q_{\mathcal{A}}}(x_i)$.
      If the sequence has no limit points then the Endpoint Theorem,
      \cite[Theorem 3.2.1]{Whale2010}, shows that there exists
      an envelopment $\phi:M\to M_\phi$ and $x\in \partial\phi(M)$
      so that $\phi(x_i)\to x$.
      Choose $\beta\in\Atlas{M_\phi}$ so that $x\in\dom(\beta)$.
      Proposition \ref{prop.inverseenvtoext} 
      implies that $(\beta\circ\phi,\ran(\beta))\in\EX{M}$.
      By construction $(\beta\circ\phi, V)\in Q_{\mathcal{A}}$ 
      where $V$ is the
      set given by Lemma \ref{lem.maxextension}. By construction
      $(\beta\circ\phi(x_i))\subset N_{(\beta\circ\phi,\ran(\beta))}$,
      $\beta(x)\in N_{(\beta\circ\phi,\ran(\beta))}$ and
      $\beta\circ\phi(x_i)\to \beta(x)$ with respect to the topology
      on $N_{(\beta\circ\phi,\ran(\beta))}$. Thus, as $q$ is continuous
      and $Q_{\mathcal{A}}(M)$ is first countable,
      $\imath_{Q_{\mathcal{A}}}(x_i)=q(\beta\circ\phi(x_i))\to q(\beta(x))$. 
      Since
      $\imath_{Q_{\mathcal{A}}}(x_i)=q(\beta(x_i))$ 
      the sequence $(\imath_{Q_{\mathcal{A}}}(x_i))$ has 
      $q(\beta(x))$ as a limit point.
    \end{proof}

    As a consequence of Proposition \ref{prop.subboundaries},
    for any arbitrary $A\subset\EX{M}$,
    the boundary $\partial\imath_A(M)$ can be viewed
    as a subset of $\ABBPM{M}$. 

    \begin{definition}\label{def.representationsofabbpm}
      Let $A\subset\EX{M}$ and let $Q_{\mathcal{A}}\subset \EX{M}$ be
      as given in Definition \ref{def:chart_boundary}.
      Let $f:A(M)\to Q_{\mathcal{A}}(M)$ be as given in Proposition 
      \ref{prop.subboundaries}.
      Let $\sigma_A$ denote the set $f(\partial\imath_A(M))$.
      This is the set of boundary points induced by $A$.
      The topological space $A(M)$ can be viewed as
      a concrete realisation of the sub-boundary
       $\sigma_A$. 
      I shall sometimes refer to
      $\sigma_A$ as the boundary induced by $A$.
      If for all sequences $(x_i)\subset M$
      the sequence $(\imath_{Q_{\mathcal{A}}}(x_i))$ has at least
      one limit point in $f(A(M))$ then the set
      $\sigma_A$ is called complete.
      If $A=\{(\alpha, U)\}$ the set $\sigma_A$ will
      be denoted by
      $\sigma_{(\alpha, U)}$. If the set $U$ is
      clear from context $\sigma_\alpha$ will be used.
    \end{definition}

    There is an important realisation of $\ABBPM{M}$ as a point
    set.

    \begin{proposition}\label{prop:ABdescibeofB}
      The boundary $\ABBPM{M}$ is in bijective correspondence with
      $\frac{\BPM{M}}{\equiv}$.
    \end{proposition}
    \begin{proof}
      This follows directly from the definitions of $\BPM{M}$ (Definition
      \ref{def.ABBPM}), 
      $Q_{\mathcal{A}}(M)$ (Definition \ref{def:chart_boundary}), 
      $\imath_{Q_{\mathcal{A}}}$ (Definition \ref{def.quotientmap})
      and $\equiv$ (Definition \ref{def.equiv}).
    \end{proof}

    Those familiar with the Abstract Boundary, \cite{ScottSzekeres1994},
    should note that $\ABBPM{M}$ has a similar structure.
    I exploit this similarity in
    Section \ref{sec.classification_chart} 
    to give a physically motivated classification
    for the elements of $\ABBPM{M}$.
    Never-the-less $\ABBPM{M}$ has a richer structure.
    In particular, $\ABBPM{M}$ comes as part of a topological space
    $Q_{\mathcal{A}}(M)$ into which $M$ is embedded, as 
    $\imath_{Q_{\mathcal{A}}}(M)$, as an open
    dense set and the equivalence relation finds expression on 
    $Q_{\mathcal{A}}(M)$
    as the coordinate transformations on the boundary which take the
    form of the homeomorphisms
    $\overline{\beta\circ\alpha^{-1}}$. These structures are not present
    in the Abstract Boundary. Refer to \cite{citeulike:9599226} for
    an example of a topology on the Abstract Boundary.

  \subsection{Boundaries induced by collections of compatible charts}
  \label{sec.compatiblecharts}

    In an ideal world a coordinate transformation on the 
    boundary, Proposition \ref{prop.coordtransonboundarydef}, 
    would actually be a coordinate transformation
    in some larger manifold. This section introduces
    a compatibility condition which is almost, but not quite
    enough to achieve this. 
    If, in addition to the compatibility 
    condition, the chosen set of extensions, $Q$, is countable
    and the extensions 
    are sufficiently nice, see Proposition \ref{prop.topmanB},
    then 
    the coordinate
    transformations on the boundary turn out
    to be genuine coordinate transformations.
    The restrictions are mild and likely to be satisfied 
    when performing computations. Thus,
    in this
    way, calculations on $\ABBPM{M}$,
    or some suitable sub-boundary, are facilitated by the compatibility 
    condition.

    \begin{definition}\label{def:contact}
      The boundary points $(\alpha,U,\{p\})$ and $(\beta,X,\{q\})$
      are in contact, $(\alpha,U,\{p\})\perp(\beta,X,\{q\})$,
      if and only if there exists $(x_i)\subset 
      \dom(\alpha)\cap\dom(\beta)$
      so that $(\alpha(x_i))\to p$ and $(\beta(x_i))\to q$. 
      The two boundary points
      $[(\alpha,U,\{p\})], [(\beta,X,\{q\})] \in \ABBPM{M}$ 
      are in contact,
      $[(\alpha,U,\{p\})]\perp[(\beta,X,\{q\})]$,
      if and only if $(\alpha,U,\{p\})\perp(\beta,X,\{q\})$. It is easy to 
      show that the 
      in contact relation is well defined on $\ABBPM{M}$.
      If $(\alpha,U,\{p\})\not\perp (\beta,X,\{q\})$ then 
      $(\alpha,U,\{p\})$ and $(\beta,X,\{q\})$ are said to be
      separate, denoted $(\alpha,U,\{p\})\parallel(\beta,X,\{q\})$.
      Likewise, if $[(\alpha,U,\{p\})] \not\perp [(\beta,X,\{q\})]$
      then $[(\alpha,U,\{p\})]$ and $[(\beta,X,\{q\})]$
      are separate, denoted $[(\alpha,U,\{p\})]\parallel[(\beta,X,\{q\})]$.
    \end{definition}
            
    The in contact relation expresses the idea of two boundary points in
    different extensions representing some similar portion of the 
    boundary of the manifold (as their identifying sets of
    sequences have non-empty intersection). This similar portion of the
    boundary is the `limit point' of the sequence $(x_i)$ used in
    Definition \ref{def:contact}.
		
    The compatibility condition is:
  
    \begin{definition}\label{def.compatible_extensions}
      Two extensions $(\alpha,U)$ and $(\beta,X)$ are 
      compatible if for all
      $p\in\BP{\alpha}\cap U$ and $q\in\BP{\beta}\cap X$,
      $(\alpha,U,\{p\})\perp (\beta,X,\{q\}) \Rightarrow 
        (\alpha,U,\{p\})\equiv (\beta,X,\{q\})$.
    \end{definition}

    Thus given two extensions a boundary point in one extension can only
    correspond to a boundary point in the other extension, not to a collection
    of boundary points. 
    A reflection of this is the following result.

    \begin{proposition}\label{lem.compatimpliehausdor}
      Let $Q\subset\EX{M}$ be a set of pair-wise compatible
      extensions. Then $Q(M)$, given in Definition \ref{def.quotientmap},
      is Hausdorff.
    \end{proposition}
    \begin{proof}
      It is trivial to show that 
      $
        \{(x,y)\in N_Q\times N_Q: q(x)=q(y)\}
      $ 
      is closed
      by the definition of $Q(M)$ and the pair-wise compatibility
      assumption. Therefore as $Q(M)$ is the image of $q$ and as
      $q$ is open, Theorem 13.10 of \cite{citeulike:593505}
      implies that $Q(M)$ is Hausdorff.
    \end{proof}

    Thus the compatibility condition ensures that the boundary points
    $\partial\imath_Q(M)$ are a nicely behaved topological space.
    All that is missing for $Q(M)$ to be a topological
    manifold with boundary
    is a suitable collection of charts and a demonstration of
    second countability. This is 
    where the assumptions of countability and of `sufficiently nice
    extensions' are required.

    \begin{proposition}\label{prop.topmanB}
      Let $Q\subset\EX{M}$ be a countable set of pair-wise compatible
      extensions so that
      for each
      $(\alpha, U)\in Q$, $N_{(\alpha, U)}$ is a
      manifold with boundary.
      Then $Q(M)$, as given in Definition \ref{def.quotientmap},
      is a topological manifold with boundary.
    \end{proposition}
    \begin{proof}
      From above $Q(M)$ is Hausdorff. It is trivial to show that
      $Q$ being a countable set implies that
      $N_Q$ is second countable. Since $q$ is open $Q(M)$ is
      second countable.
      Let 
      \begin{multline*}
        A=\{\alpha\circ\imath_Q^{-1}:\alpha\in\Atlas{M}\}\ \cup\\
          \bigcup_{(\alpha, U)\in Q)}
          \{f\circ (q|_{N_{(\alpha,
          U)}})^{-1}:(\alpha, U)\in Q,\,f\in\Atlas{N_{(\alpha,U)}}\}.
        \end{multline*}
      The collection of maps in $A$ gives the needed
      homeomorphisms to $\R^{n-1}\times\{x\in\R:x\geq 0\}$.
      Therefore $Q(M)$ is a topological manifold with boundary.
    \end{proof}
    
    The transition functions on the boundary
    $\overline{\beta\circ\alpha^{-1}}$ now play the role of genuine
    coordinate transformations as 
    $f\circ(q|_{N_{(\alpha, U)}})^{-1}\circ q|_{N_{(\beta, V)}}\circ h^{-1}=
    f\circ\overline{\beta\circ\alpha^{-1}}\circ h^{-1}$
    for $f\in\Atlas{N_{(\alpha, U)}}$ and 
    $h\in \Atlas{N_{(\beta, V)}}$.

    If one is willing to assume, for any pair of
    extensions $(\alpha, U),(\beta,V)\in Q$ that the extensions of
    the transition functions $\beta\circ\alpha$
    to some subset of $\partial_U\ran(\alpha)$ is $C^k$
    then the following corollary holds.

    \begin{corollary}\label{cor:propenvindbound}
      Let $M$ be a $C^l$ manifold.
      If $Q\subset\EX{M}$ is a set of pair-wise compatible
      extensions and there exists $k\in\N$, $0\leq k\leq l$, so that;
      \begin{enumerate}
        \item 
          for all
          $(\alpha, U)\in Q$, $N_{(\alpha, U)}$ is a
          $C^k$,  manifold with boundary,
        \item for all $(\alpha, U),(\beta, V)\in Q$ the
          functions $\overline{\beta\circ\alpha^{-1}}$ are $C^k$.
      \end{enumerate}
      then $Q(M)$ is a $C^k$ manifold with boundary.
    \end{corollary}
    \begin{proof}
      Let $A$ be as in the proof of Proposition \ref{prop.topmanB}.
      We need to check that the transition functions between elements
      of $A$ are $C^k$. This will demonstrate that
      $A$ generates a $C^k$ atlas for $Q(M)$.
      We have three cases.
      \textbf{Case 1.} Let $f\circ q|_{N_{(\alpha, U)}}^{-1}$
      and $g\circ q|_{N_{(\beta, U)}}^{-1}$ be in $A$. The
      transition map, by Proposition \ref{prop.coordtransonboundarydef},
      is $g\circ\overline{\beta\circ\alpha^{-1}}\circ f$. This is
      $C^k$ by assumption. \textbf{Case 2.}
      Let $f\circ q|_{N_{(\alpha, U)}}^{-1}$ and $\beta\circ\imath_Q$
      be in $A$. There are two possible transition functions, 
      $f\circ q^{-1}|_{N_{(\alpha, U)}}\circ \imath_Q\circ\beta^{-1}$
      and
      $\beta\circ\imath_Q^{-1}\circ q|_{N_{(\alpha, U)}}\circ f^{-1}$.
      It is easy to show that  
      $\beta\circ\imath_Q^{-1}\circ q|_{N_{(\alpha,
      U)}}=\overline{\beta\circ\alpha^{-1}}$
      and
      $q^{-1}|_{N_{(\alpha, U)}}\circ \imath_Q\circ\beta^{-1}
      =\overline{\alpha\circ\beta^{-1}}$. 
      Thus the transition maps amount to $f\circ\overline{\alpha\circ\beta^{-1}}$ 
      and $\overline{\beta\circ\alpha^{-1}}\circ f^{-1}$
      which are both $C^k$.
      \textbf{Case 3.} Let $\alpha\circ\imath_Q^{-1}$ and
      $\beta\circ\imath_Q^{-1}$ be in $A$. In this case
      the transition functions are
      $\beta\circ\alpha^{-1}$  and $\alpha\circ\beta^{-1}$
      which
      are $C^l$.
      Thus $Q(M)$ equipped with 
      the atlas generated by $A$ is a $C^k$ manifold with boundary.
    \end{proof}

    There is a converse to Proposition \ref{prop.topmanB} and
    its corollary.

    \begin{proposition}\label{prop.Q(M)fromphi}
      Let $\phi:M\to M_\phi$ be an envelopment.
      Then, in the notation of Proposition \ref{prop.inverseenvtoext},
      $
        Q=\{(\alpha\circ\phi,\ran(\alpha)):\alpha\in\Atlas{M_\phi},\,
          \dom(\alpha)\cap\partial\phi(M)\neq\varnothing\}
      $
      is a pairwise compatible set of extensions and
      there exists a homeomorphism, in the notation
      of Definition \ref{def.quotientmap},
      $f:Q(M)\to\overline{\phi(M)}$       
      so that $f\circ\imath_Q=\phi$,
      $f(q(N_{(\alpha\circ\phi)}))=\dom(\alpha)\cap\overline{\phi(M)}$
      and $\overline{\beta\circ\phi\circ\phi^{-1}\circ\alpha^{-1}}=
      \beta\circ\alpha^{-1}$.
    \end{proposition}
    \begin{proof}
      The proof is long and depends on results in
      Appendix \ref{sec.topQ(M)} so the proof
      is given in Appendix \ref{sec.compcompexprop}.
    \end{proof}

    The equation $\overline{\beta\circ\phi\circ\phi^{-1}\circ\alpha^{-1}}=
    \beta\circ\alpha^{-1}$ again shows that the coordinate transforms
    on the boundary are genuine coordinate transformations.

    \begin{corollary}\label{cor.Q(M)isamanifoldifenvboundaryman}
      Let $\phi:M\to M_\phi$ be an envelopment and let $Q(M)$
      be as given in Definition \ref{def.quotientmap}.
      If $\overline{\phi(M)}$ is a manifold
      with boundary then $Q(M)$ is a manifold with boundary and
      $f$ is a diffeomorphism, where $f$ is given in Proposition 
    \ref{prop.Q(M)fromphi}.
    \end{corollary}
    \begin{proof}
      The proof depends on the Proposition \ref{prop.Q(M)fromphi}
      and so has been moved to Appendix \ref{sec.compcompexprop}.
    \end{proof}

    Proposition \ref{prop.Q(M)fromphi} has an interesting consequence.
    Given any envelopment $\phi:M\to M_\phi$ the topological
    boundary $\partial\phi(M)$
    is a set of representatives for the 
    set of boundary points $\partial\imath_Q(M)$, where
    $Q$ is as given in Proposition \ref{prop.Q(M)fromphi}.
    Thus, via Definition \ref{def.representationsofabbpm}, 
    any `boundary' of the manifold, that can be presented as 
    the topological boundary of an envelopment,
    induces a set of boundary points $\sigma_Q\subset\ABBPM{M}$ 
    and therefore the boundary given by the embedding
    is 
    contained in $\ABBPM{M}$.

    A notable example of a boundary defined in this way
    is Penrose's conformal boundary. Thus it is in precisely the
    manner described above that $\ABBPM{M}$ generalises the conformal
    boundary. In Section \ref{sec:Examples} I present the concrete 
    calculation of
    $\sigma_Q$ for the conformal boundary of Minkowski spacetime.

    Lastly, note that Proposition \ref{prop.Q(M)fromphi}
    did not require the assumption that
    $N_{(\alpha\circ\phi,\ran(\alpha))}$ is a manifold
    with boundary nor did it require that $Q$ was countable, in contrast
    to Proposition \ref{prop.topmanB}. 
    Thus the results of this
    section are not sharp and it might be possible to
    improve on Proposition \ref{prop.topmanB}.
    Never-the-less Proposition \ref{prop.topmanB} will, in most cases, be
    sufficient, so the generalisation of this proposition has not
    been sort.
  
  \subsection{Classification of boundary points}\label{sec.classification_chart}

    Given an envelopment, $\phi:M\to M_\phi$, of a manifold and
    a metric or connection on $M$, Scott and Szekeres, 
    \cite[Section 4]{ScottSzekeres1994}, have presented
    a classification of the points in $\partial\phi(M)$ into a hierarchy
    of physically motivated classes and studied the invariance of
    these classes under the construction of the Abstract Boundary.
    
    Proposition \ref{lem:embedding given by an extension} 
    shows that each extension
    $(\alpha,U)$ induces an envelopment, 
    $\Psi(\alpha,U):M\to M_{(\alpha,U)}$, so that 
    $\BP{\alpha}\cap U=\partial\Psi(\alpha,U)(M)$. Hence  Scott
    and Szekeres' classification of the elements of
    $\partial\Psi(\alpha,U)(M)$ can be used to classify the elements of 
    $\BP{\alpha}\cap U$, as long as there is a metric or connection on 
    $M$. 
    
    Once the classification has been defined, its invariance under the
    equivalence relation used to define $\ABBPM{M}$ is then discussed. 
    Due to the similarity of the Abstract Boundary and $\ABBPM{M}$
    the invariance of the classification under coordinate transformations
    on the boundary is exactly analogous to the invariance of the
    classification
    proved in Scott and Szekeres' paper \cite{ScottSzekeres1994}.
    The result is a physical classification of $\ABBPM{M}$.
    With this in mind, the rest of this section follows the
    classification of Abstract Boundary points, 
    \cite[Section 4]{ScottSzekeres1994}, closely. The discussion below is 
    restricted to the case of a metric on $M$, 
    all results, however, can be extended in an analogous manner
    to the case of a connection on $M$.

    The classification depends on a choice of curves.
    \begin{definition}[{\cite[Definition 4]{ScottSzekeres1994}}]\label{def.BPP}
      A set $\mathcal{C}$ of parametrized curves in $M$ is said to have 
      the bounded parameter property (b.p.p.)\
      if the following conditions are satisfied;
      \begin{enumerate}
        \item for all $p\in{M}$ there exits $\gamma\in\mathcal{C}$ 
          so that $p\in\gamma$,\label{def.BPP.1}
        \item if $\gamma\in\mathcal{C}$ and $\delta$ is a sub-curve of 
          $\gamma$ then $\delta\in\mathcal{C}$,\label{def.BPP.2}
        \item for all $\gamma,\delta\in\mathcal{C}$, if $\delta$ is 
          obtained from 
        $\gamma$ by a change of parameter
          then either both curves are bounded or 
          both are unbounded.\label{4.3}\label{def.BPP.3}
      \end{enumerate}
    \end{definition}
    
      Given $M$ a manifold and $(\alpha,U)\in\EX{M}$ the 
      classification begins by 
      dividing the elements of $\BP{\alpha}\cap U$ into four classes.

      \begin{definition}\label{def:chart_regular}
        Let $(M,g)$ be a $C^l$ pseudo-Riemannian manifold.
        The boundary point $(\alpha,U,\{p\})$ is $C^k$ regular if
        there exists an open neighbourhood $V\subset U$ of $p$
        and a $C^k$ pseudo-Riemannian metric $\hat{g}$ on $V$
        so that $(V,\hat{g})$ is an extension of $(\ran(\alpha)\cap V,
        (\alpha^{-1})^*\left(g|_{\alpha^{-1}\left(\ran(\alpha)\cap V\right)}\right))$ as a
        manifold with pseudo-Riemannian metric. If $(\alpha,U,\{p\})$ is not
        $C^k$ regular then it is $C^k$ non-regular. 
        The reference to $\alpha$ and $U$ will be dropped if clear from context.
        When the degree of
        regularity is unimportant it too will be dropped. 
      \end{definition}
      
      \begin{definition}\label{def:chartapproach}
				The boundary point $(\alpha, U, \{p\})$ is approachable if
        there exists $\gamma\in\mathcal{C}$ so that 
        $p$ is in the closure
        of $\alpha\circ\gamma|_{\gamma^{-1}
        \left(\gamma\cap\dom(\alpha)\right)}$. The point $p$
        is unapproachable if it is not approachable.
        The reference to $\alpha$ and $U$ will be dropped if clear from context.
      \end{definition}

      The four classes, mentioned above, are approachable regular 
      points, approachable non-regular 
      points, unapproachable regular points and unapproachable non-regular 
      points.
      With these classes defined, and with the relation $\covers$, the
      classification of boundary points can follow the classification
      in \cite{ScottSzekeres1994} without change. 
      
      By virtue of the choice of curves (which should be taken 
      as large as possible, e.g.\ the set of all curves 
      with generalized affine parameter \cite{citeulike:4251705}) 
      unapproachable points are considered to
      be physically unimportant, even though unapproachable regular points
      may exist. 
      
      \begin{example}
        Continuing from Example \ref{ex:sincurve}. 
        Let $g$ be the standard
        metric on $\R^2$ then $h=g|_{TM\times TM}$
        is a metric on $M$ so that $(\R^2,g)$ is an extension of $(M,h)$.
        Let $\mathcal{C}$ be the set of all affinely parametrised geodesics
        with respect to $h$. The set $\mathcal{C}$ satisfies the bounded
        parameter property.
        With respect to this choice of curves
        $(\alpha,U,\{(0,1)\})$ is a non-approachable regular boundary point.
      \end{example}
      
      Approachable regular points are points through which the manifold can be 
      extended. Approachable non-regular points are elements of the boundary
      that are physically relevant and through which the manifold cannot
      be extended. They    
      can be further subdivided into 
      points at infinity and singular points. 
    
      \begin{definition}\label{def:chart_inf_bp}
        Let $\mathcal{C}$ be a set of curves
        with the b.p.p.
        A boundary point $(\alpha, U, \{p\})$ is a $C^k$ point at infinity if; 
        \begin{enumerate}
          \item $p$ is not a $C^k$ regular boundary point,
          \item $p$ is approachable,
          \item for all $\gamma\in \mathcal{C}$, if $\gamma$ approaches $p$
            then $\gamma$ is unbounded.
        \end{enumerate}
        The reference to $\alpha$ and $U$ will be dropped if clear from context.
        When the degree of
        regularity is unimportant it too will be dropped. Thus
        a $C^k$ point at infinity will sometimes be referred to as a point at infinity.
      \end{definition}
  
      \begin{definition}\label{def:chart_sing_bp}
        Let $\mathcal{C}$ be a set of curves
        with the b.p.p.
        A boundary point $(\alpha,U,\{p\})$ is a $C^k$ singularity if; 
        \begin{enumerate}
          \item $p$ is not a $C^k$ regular boundary point,
          \item there exists $\gamma\in\mathcal{C}$ so that $\gamma$ approaches
            $p$ with bounded parameter.
        \end{enumerate}
        The reference to $\alpha$ and $U$ will be dropped if clear from context.
        When the degree of
        regularity is unimportant it too will be dropped. Thus
        a $C^k$ singularity will sometimes be referred to as a singularity.
      \end{definition}
      
      Approachable non-regular points may be non-regular due to properties of
      the chart rather than any inherent property of the manifold.     
    
      \begin{definition}
        A point at infinity
        $(\alpha,U,\{p\})$
        is removable if there exists 
        a boundary set $(\beta, X, Y)$
		    so that $(\beta,X,Y)\covers (\alpha,U,\{p\})$ and for all $y\in Y$, 
		    $(\beta, X, \{y\})$ is a
        regular boundary point.
              
        A singular point $(\alpha,U,\{p\})$
        is removable if there exists 
        a boundary set $(\beta, X, Y)$
		    so that $(\beta,X,Y)\covers (\alpha,U,\{p\})$ and for all $y\in Y$, 
		    $(\beta, X,\{y\})$ is a
        regular boundary point, a point at infinity or unapproachable.
        
        A non-regular boundary point that is either a removable point at infinity
        or a removable singularity will be called a removable boundary point.
        A boundary point that is not removable will be called essential.
      \end{definition}

      The definition of a removable boundary point
      $(\alpha, U, \{p\})$, covered by the boundary
      set $(\beta, X, Y)$ (so that for all $y\in Y$, $(\beta, X, \{y\})$ is
      a regular boundary point if $p$ is a point at infinity or
      $(\beta, X, \{y\})$ is a regular point, unapproachable point
      or point at infinity if $p$ is a singular point)
      is such that the chart $\beta$ resolves the non-regular behaviour
      of $p$ by representing $[p]\in\ABBPM{M}$ by the set of boundary points
      $\{[y]: y\in Y\}$. Thus $\beta$ removes the poor behaviour of $\alpha$.
      
      Approachable non-regular points may cover a regular 
      boundary point of some other extension. Such approachable non-regular 
      points will necessarily exhibit directional behaviour.
      That is, given a curve whose endpoint is such a point at infinity 
      or singularity, certain physical quantities, e.g.\ the Kretschmann scalar,
      could have either regular or singular behaviour.
      
      The Curzon singularity is presented in 
      Section \ref{sec:curzon}, see
      \cite{Scott1986CurzonI,Scott1986CurzonII}. 
      This is a well known directional
      singularity. The Curzon singularity covers regular points,
      points at infinity and unapproachable points once 
      the singularity is expressed in a suitable
      chart. 
   
      \begin{figure*}
		    \centering
		    \def\svgwidth{\columnwidth}
		    \input{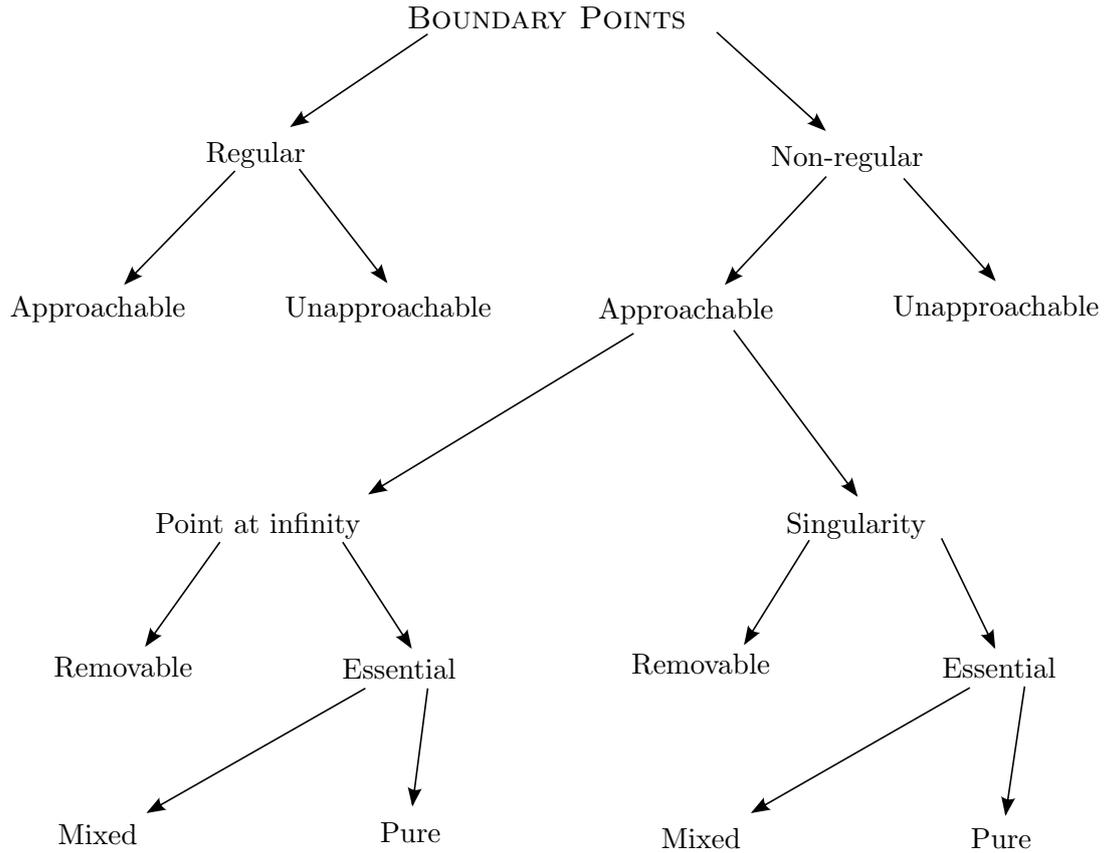}
		    \caption{A graphic representation of the classification given
		    in Section \ref{sec.classification_chart}}\label{fig:2}
		  \end{figure*}
     
      \begin{definition}\label{def.mixed}
        An essential boundary point $(\alpha,U,\{p\})$ is
        mixed, or directional, if there exists 
        a regular boundary point $(\beta, X, \{q\})$
		    so that $(\alpha,U,\{p\})\covers (\beta,X,\{q\})$.
        An essential boundary point that is not mixed is called
        pure.
      \end{definition}  
    
      Figure \ref{fig:2} gives a summary of the classification.
    
      Extensions that include mixed or removable points are undesirable
      since this behaviour is considered to be 
      an artefact of the particular chart used. It is a long standing 
      problem in General Relativity to show that all directional
      singularities can be expressed, as in the case of the Curzon solution
      \cite{Scott1986CurzonI,Scott1986CurzonII}, as a collection of genuine
      singularities, regular points, points at infinity 
      and unapproachable points. 
      In our context this
      problem can be phrased as: does every manifold have a complete
      boundary that does not contain removable or mixed boundary points?

      As mentioned above
      invariance under $\equiv$ is necessary before
      any quantity defined on the boundary $\BP{\alpha}\cap U$ of an
      extension $(\alpha,U)\in\EX{M}$ can be considered to be coordinate 
      invariant.
      It turns out that not all the classes of the classification of 
      $\BP{\alpha}\cap U$ are invariant under $\equiv$.
      
      \begin{proposition}\label{prop.App&UnAppABBPM} 
        The following classification of $\ABBPM{M}$ is well defined.
			  Let
			  $\cur{C}$ be a set of curves with the b.p.p.
			  A boundary point
			  $[(\alpha, U, \{p\})]\in\ABBPM{M}$ is;
			  \begin{enumerate}
			    \item approachable if and only if $(\alpha, U, \{p\})$ is 
            approachable,
			    \item a $C^k$ indeterminate boundary point
			      if and only if 
            $(\alpha,U,\{p\})$ is a $C^k$ 
            regular point, a $C^k$ removable point at infinity
			      or a $C^k$ removable singularity,
			    \item is a $C^k$ point at infinity if and only if
			      $(\alpha,U, \{p\})$ is an essential $C^k$point at infinity,
			    \item is a $C^k$ singularity if and only if
			      $(\alpha, U, \{p\})$ is an essential $C^k$ singularity.
			  \end{enumerate}
			  Furthermore, 
			  an essential boundary point $[(\alpha, U, \{p\})]$
			  is mixed, or directional, if
			  and only if $(\alpha, U,\{p\})$ is mixed.
			  The boundary point $[(\alpha, U, \{p\})]$ is pure
			  if and only if $(\alpha, U, \{p\})$ is pure.
		  \end{proposition}
      \begin{proof}
        The proof of this is very long as it, essentially, involves
        checking each case individually. Moreover the full proof
        is virtually identical to that give in \cite{ScottSzekeres1994},
        hence I give only a sketch of the result.

        Definitions \ref{def:chart_regular} and \ref{def:chartapproach}
        correspond, via Propositions \ref{lem:embedding given by an extension}
        and \ref{prop.inverseenvtoext}, to the Abstract Boundary
        definitions of a regular point and an approachable
        point, (\cite[Definitions 28 and 12]{ScottSzekeres1994}),
        respectively.
        Moreover, again via 
        Propositions \ref{lem:embedding given by an extension}
        and \ref{prop.inverseenvtoext}, the boundary point
        $(\alpha, U, \{p\})$ covers the boundary
        point $(\beta, V, \{q\})$ if and only if the point
        $p\in\partial\Psi(\alpha, U)(M)$ covers
        $q\in\partial\Psi(\beta, V)(M)$ according to the
        Abstract Boundary definition of covers, 
        \cite[Definition 14]{ScottSzekeres1994} (see also
        \cite[Theorem 19]{ScottSzekeres1994}).
        It is now clear that the definitions of the various classes of the 
        classification given here correspond
        to the definitions given in \cite{ScottSzekeres1994}.
        Thus the results of Table 1 of \cite{ScottSzekeres1994}
        are valid for the classifications defined here.
        Table 1 gives, for each pair of classes, the
        possibility for the first class to the covered by the second
        class.
        Hence, via Table 1, it is possible to determine the invariance
        of each class under the covering relation.
        From above it is clear that the same results as for the 
        Abstract Boundary apply to the classification
        given here, thus the result is
        proved, see \cite[Section 5]{ScottSzekeres1994}.
      \end{proof}
      
      Figure \ref{fig:3} gives a summary of the coordinate 
      independent classification.
      
		  \begin{figure*}
			  \centering
			  \def\svgwidth{\columnwidth}
			  \input{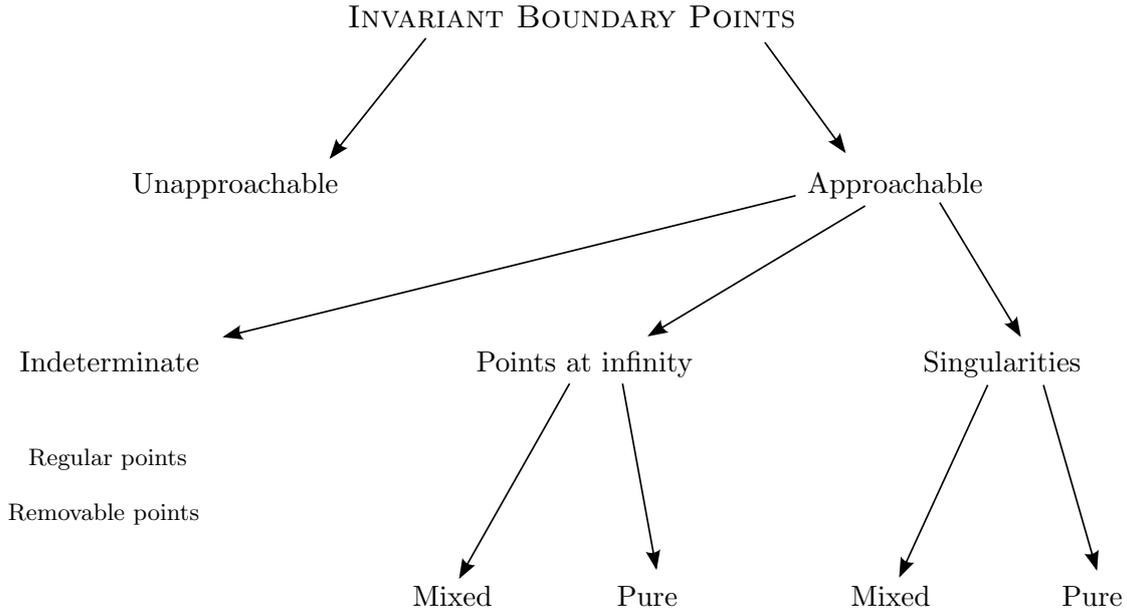}
			  \caption{A graphic representation of the coordinate 
			  invariant classification given
		    in Section \ref{sec.classification_chart}}\label{fig:3}
		  \end{figure*}
 
\section{Examples}\label{sec:Examples}

  This section provides three examples of global analysis
  which I consider typical of the chart based approach. 
  In each case only material
  from the sources is used to construct and classify a complete representation
  of the boundary. It is hoped, therefore, that these examples will 
  demonstrate to the reader that
  Section \ref{sec:MAIN_chart_boundary} presents a useful
  formalization of `the chart based approach to studying global structure'.

  The first example is the construction of Penrose's conformal boundary
  for Minkowski spacetime, as it is given in
  \cite[Section 5.1]{Hawking1975Large} (note that this presentation of 
  the boundary differs from Penrose's own publications
  \cite{citeulike:8978534,citeulike:9600515,citeulike:10075981}, in particular
  it
  includes spacelike and
  timelike infinity).
  The charts 
  used to produce the conformal
  boundary also induce a boundary, 
  as in Definition \ref{def.representationsofabbpm},  
  so that every point of this boundary is a pure point at infinity.
  This example is presented in great detail in order to provide examples
  of compatible charts, the construction of a boundary corresponding
  to collections of compatible charts and how coordinate singularities can 
  be handled.
  
  The second example is the construction of the maximal extension of
  the Curzon solution given by Scott and Szekeres, 
  \cite{Scott1986CurzonI,Scott1986CurzonII}.
  The focus of this example is on
  how different charts can induce boundaries with different properties.
  Thus issues arising from coordinate singularities and
  restrictions of coordinates are ignored.
  An example of a boundary containing a directional singularity 
  is produced along with a demonstration of 
  how a clever choice of coordinates can `resolve' the directional
  singularity. The boundary induced by this `better' chart 
  contains regular points,
  pure points at infinity, pure singularities and unapproachable points.
  
  The third example is the analysis of the global structure 
  of smooth Gowdy symmetric
  generalized Taub-NUT solutions with spatial $\Sph^3$-topology
  given by Beyer and Hennig, \cite{Beyer2011Smooth}.
  The spacetimes are presented as the solutions to a system of partial 
  differential equations (PDEs).
  A closed form for the metric of these space times is not given.  
  Using 
  Beyer and Hennig's analysis of these equations it is possible to show, 
  however, that the boundary of these spacetimes can contain regular points
  and essential singularities. From the analysis of \cite{Beyer2011Smooth}
  it is not possible to conclude that the singularities are pure.

  The 
  $c$-boundary \cite{Flores2010Final} and the causal boundary 
  \cite{citeulike:8664207}, require knowledge of the causal structure
  for their construction.
  The chart based global analysis of the second and third examples, however,
  do not determine the causal structure across the whole manifold.
  In the case of the second example, the
  majority of the analysis is performed using collections of spacelike
  curves. In particular the geodesic equations are not solved in a neighbourhood
  of the singularity and hence additional work would be required before using
  either the $c$-boundary or the causal boundary to 
  analyse the global structure of 
  the singularity.
  In the case of the third example 
  the properties of the metric are studied on particular curves along which
  the PDEs reduce to ODEs. Not 
  only is the causal structure on the manifold not discussed, but exact
  solutions of the PDEs are not given except on these curves.
  Thus additional work would be required before using 
  either the $c$-boundary or the causal boundary to analyse the 
  global structure of solutions to the PDEs. These comments also apply
  to the $g$ and $b$-boundaries.
  In this sense $\ABBPM{M}$ is an  easier construction to work
  with.

  \subsection{Minkowski spacetime}\label{sec:Minkowski}
  
  This example follows the coordinate construction of 
  Penrose's conformal boundary for
  Minkowski spacetime as given in  \cite[Section 5.1]{Hawking1975Large} (see
  \cite{citeulike:8978534,citeulike:9600515,citeulike:10075981} for 
  Penrose's own presentations, which differ slightly)
  and shows that the charts used to construct the conformal
  boundary also admit extensions that give a representation of Penrose's
  conformal boundary within $\ABBPM{M}$, see Section
  \ref{sec.compatiblecharts}.
  The standard analysis of affinely parametrised
  geodesics in Minkowski spacetime implies that every element
  of this induced boundary is a pure point at infinity. 
  
    Minkowski spacetime is $M=\R^4$ equipped with the metric
    $
      ds^2=-dt^2+dx^2+dy^2+dz^2
    $
    for $(t,x,y,z)\in\R^4$. This chart gives no information about 
    the boundary of the spacetime as its domain is complete, see Definition
    \ref{def:BP}.
    In order to introduce admissible boundary points it is necessary to 
    find charts whose images in $\R^4$ have non-empty topological boundaries.
    The construction of $\ABBPM{M}$ provides no information about which
    charts to use. Indeed $\ABBPM{M}$ is 
    chart agnostic. Two obvious examples are the one-point compactification
    of $M$
    and an embedding of $M$ into the $4$-dimensional unit ball. 
    Both of these charts would induce complete boundaries, but neither
    is suitable for the material of this section. Here, I am showing how
    the constructions of the previous section relate to Penrose's
    conformal compactification of Minkowski space. 
    Hence the charts used must behave appropriately with respect to
    this conformal structure.
    Therefore 
    the conformal structure
    of Minkowski space is an additional piece of 
    physical information that is being
    used to provide guidance on the construction of the 
    charts of this section. In general, when working
    with an arbitrary manifold, the selection of appropriate charts will
    require geometric or physical insight, just as is currently required
    for the chart based approach to the study of global structure.
    
    Define a new chart by
    \begin{align*}
      \alpha_1(t,x,y,z)&=\left(t,\,\sqrt{x^2+y^2+z^2},\,
				\cos^{-1}\left(\frac{z}{\sqrt{x^2+y^2+z^2}}\right),\,\arctan(x,y)\right)\\
      &=(t,r_1,\theta_1,\phi_1),
    \end{align*}
    where $\arctan:\R\times\R\to(-\pi,\pi)$ is as given in Example \ref{ex:bp_ex}. 
		Using $t,x,y,z$ coordinates, let
    $O=\left\{(t,0,0,0):t\in\R\right\}$,
    $X=\left\{(t,x,0,z):t,z\in\R,\, x<0\right\}$
    and $Z = \left\{(t,0,0,z):t\in\R,z\in\R, z\neq 0\right\}$.
    In order to avoid coordinate singularities, and because of the
    requirement that both $\dom(\alpha_1)$ and $\ran(\alpha_1)$
    be open, it is necessary to take
    $
      \dom(\alpha_1)=\R^4\setminus
        \left(O\cup X\cup Z\right)
    $
    and 
    $
      \ran(\alpha_1)=\R\times\R^+\times(0,\pi)\times(-\pi,\pi).
    $  
    To cover the whole manifold, take a second chart 
    using spherical coordinates centred on 
    $(0,1,1,0)$ with the $z$ and $y$ axis swapped and the $x$-axis taken 
    to $-x$-axis. Note that these are the spherical coordinates based on a 
    translated and rotated Cartesian frame with respect to 
    the frame given by the $t,x,y,z$ coordinates. 
    This ensures that the new chart appropriately resolves the 
    points in $M$ that are not in $\dom(\alpha_1)$ as $r_1\to\infty$.
    The equation for the new chart is
    \begin{align*}
      \alpha_2(t,x,y,z)&=
        \biggl(t,\,\sqrt{(x-1)^2+(y-1)^2+z^2},\\
					&\hspace{1.4cm} \cos^{-1}
          \left(\frac{y-1}{\sqrt{(x-1)^2+(y-1)^2+z^2}}\right),\,
          \arctan(1-x,z)\biggr)\\
        &=(t,r_2,\theta_2,\phi_2).
    \end{align*}
    Similarly to $\alpha_1$,
    \begin{multline*}
      \dom(\alpha_2)=\R^4\setminus
        \biggl(\left\{(t,1,1,0):t\in\R\right\}\\
        \cup\left\{(t,x,y,0):t,y\in\R,\, x>1\right\}
        \cup\left\{(t,1,y,0):t\in\R,\,y\in\R,\,y\neq 0\right\}\biggr)
    \end{multline*}
    and 
    $
      \ran(\alpha_2)=\R\times\R^+\times(0,\pi)\times(-\pi,\pi).
    $
		
    The topological boundary, $\partial\ran(\alpha_i)$, of $\ran(\alpha_i)$, $i\in\{1,2\}$, in $\R^4$ is the 
    union of the disjoint sets
    \begin{gather}
      \R\times\{0\}\times[0,\pi]\times[-\pi,\pi]\label{set:min1}\\
      \R\times\R^+\times[0,\pi]\times\{-\pi,\pi\}\label{set:min2}\\
			\R\times\R^+\times\{0,\pi\}\times[-\pi,\pi]\label{set:min3}.
			\end{gather}    
    A sequence $(x_k)\subset \dom(\alpha_i)$, $i\in\{1,2\}$,
    so that $(\alpha_i(x_k))$ converges
    to;
    \begin{itemize}
      \item a point in set \eqref{set:min1} 
				will converge to a point in 
        $O$, if $i=1$, or, where $(t,x,y,z)$ coordinates are used,
        $\{(t,1,1,0):t\in\R\}$, if $i=2$,
      \item a point in set \eqref{set:min2} will converge to a point in 
        $X$, if $i=1$, or, where $(t,x,y,z)$ coordinates are used,
        $\left\{(t,x,y,0): t, y\in\R,\, x>1\right\}$, if $i=2$,
      \item a point in set \eqref{set:min3}
				will converge to a point in 
        $Z$, if $i=1$, or, where $(t,x,y,z)$ coordinates are used,
        $\left\{(t,1,y,0):t\in\R,y\in\R,\,y\neq0\right\}$, if $i=2$.
    \end{itemize}
    Thus no point in the union
    of the sets \eqref{set:min1},
    \eqref{set:min2} and \eqref{set:min3} is an admissible boundary point.
    
    Some form of 
    compactification of the coordinates will introduce 
    admissible boundary points. As mentioned before the
    guiding principle in this section is the conformal structure
    of $M$. Since $t - r$ and $t + r$ are coordinates with
    null tangent vectors they have a form of invariance with respect
    to conformal transformations of $M$. 
    Hence rather than compactifying $t$ and $r$, a change of coordinates
    to 
    $t + r$ and $t - r$ will be performed and then compactification 
    of these new coordinates will be done.
    The result is below. Before I present this,
    it is worth pointing out one
    particular consequence of this choice.
    Because null geodesics in Minkowski space are necessarily traveling forward
    or backward in time, ``infinity'' for each fixed $t$ slice of $M$ will
    be identified on the boundary. Below, this identified point is
    denoted
    $\imath_i^0$. There are charts that respect the conformal structure
    of $M$ and which do not perform this identification, 
    see for example
    \cite{citeulike:8973162,citeulike:11043341}, but they shall not be
    needed here. As I have mentioned before, 
    there is great freedom in the choice of what chart to use to build a local
    representation of the boundary of $M$. Because of this additional
    information, e.g.\
    geometric or 
    physical, something beyond the purely topological is needed to inform any choice.

    Let $i\in\{1,2\}$ and
    define 
    \[
      \beta_i(t,r_i,\theta_i,\phi_i)=(\tan^{-1}(t+r_i),\tan^{-1}(t-r_i),
        \theta_i,\phi_i)=(p_i,q_i,\theta_i,\phi_i).
    \]
    Hence $\dom(\beta_i)=\dom(\alpha_i)$ and
    \[
      \ran(\beta_i)=\biggl\{(p_i,q_i):p_i,q_i\in
        \left(-\frac{1}{2}\pi,\frac{1}{2}\pi\right),\,q_i\leq p_i\biggr\}
        \times (0,\pi)\times(-\pi,\pi).
    \]
    The boundary, $\partial\ran(\beta_i)$, 
    of $\ran(\beta_i)$, $i\in\{1,2\}$, in $\R^4$
    is given by the union of the sets
    \begin{gather}
      \biggl\{(p_i,q_i):p_i,q_i\in
        \left(-\frac{1}{2}\pi,\frac{1}{2}\pi\right),\,q_i\leq p_i\biggr\}\times
        \{0,\pi\}\times[-\pi,\pi]\label{set:min4},\\
      \biggl\{(p_i,q_i):p_i,q_i\in
        \left(-\frac{1}{2}\pi,\frac{1}{2}\pi\right),\,q_i\leq p_i\biggr\}\times
        [0,\pi]\times\{-\pi,\pi\}\label{set:min5},
    \end{gather}
    and the sets
    \begin{gather*}
      \mathcal{J}^+_i=\left\{\frac{1}{2}\pi\right\}\times
        \left(\frac{-1}{2}\pi,\frac{1}{2}\pi\right)
        \times(0,\pi)\times(0,2\pi),\\
      \mathcal{J}^-_i=\left(\frac{-1}{2}\pi,\frac{1}{2}\pi\right)\times
        \left\{\frac{-1}{2}\pi\right\}
        \times(0,\pi)\times(0,2\pi),\\
      \imath^+_i=\left\{\frac{1}{2}\pi\right\}\times
        \left\{\frac{1}{2}\pi\right\}\times(0,\pi)\times(0,2\pi),\\
      \imath^0_i=\left\{\frac{1}{2}\pi\right\}\times
        \left\{\frac{-1}{2}\pi\right\}\times(0,\pi)\times(0,2\pi),\\
      \imath^-_i=\left\{\frac{-1}{2}\pi\right\}\times
        \left\{\frac{-1}{2}\pi\right\}\times(0,\pi)\times(0,2\pi).
    \end{gather*}
    As before the elements of the sets \eqref{set:min4} and \eqref{set:min5} 
    are not admissible 
    boundary points. The other sets correspond to the usual
    expression of the conformal boundary of Minkowski space with respect
    to the chart $\beta_i$.
    The sets have been labelled so that this correspondence is obvious, e.g.\
    future null infinity is represented by 
    $\mathcal{J}^+_1$ and $\mathcal{J}^+_2$ for $\beta_1$ and $\beta_2$
    respectively, see 
    \cite[page 123]{Hawking1975Large}.

    Since the ranges of $\beta_1$ and $\beta_2$ are the same, they have
    have the same admissible boundary points, e.g.\ 
    $\mathcal{J}^+_1=\mathcal{J}^+_2$. 
    It is important, however, to maintain
    a distinction between these admissible boundary points since the domains
    of $\beta_1$ and $\beta_2$ are different and therefore their 
    admissible boundary points
    represent different parts of the boundary of $M$.
    Letting $W=\R\times\R\times(0,\pi)\times(-\pi,\pi)$
    the extension
    $(\beta_i,W)\in\EX{M}$, for $i\in\{1,2\}$, is such that 
    $\partial_W\ran(\beta_i)$ is the union
    of the sets $\mathcal{J}^+_i,\mathcal{J}^-_i,\imath^+_i,\imath^0_i$ and 
    $\imath^-_i$. Thus
    each extension, $(\beta_i,W)$, induces a boundary, $\sigma_i$,
    where $[(\beta_i,W,\{p\})]\in\sigma_i$
    if and only if $p$ is an element of one of 
    $\mathcal{J}^+_i,\mathcal{J}^-_i,\imath^+_i,\imath^0_i$ and $\imath^-_i$. 
    
    The two sets $\sigma_1$ and $\sigma_2$ are not equal.    
    For example, since $X$ is not in the domain of $\beta_1$ any sequence 
    lying in 
    $X$ will not converge to a boundary point in $\sigma_1$.
    However, $X$
    is in the domain of $\beta_2$ thus there are boundary
    points representing the `limit points of $X$ at infinity' in $\sigma_2$.
    Since 
    \begin{multline*}
      \beta_2(X)=\biggl\{
        \biggl(\tan^{-1}\left(t+\sqrt{(x-1)^2+1+z^2}\right),\,\\
				\tan^{-1}\left(t-\sqrt{(x-1)^2+1+z^2}\right),\,\\
				\cos^{-1}\left(\frac{-1}{\sqrt{(x-1)^2+1+z^2}}\right),\,
				\arctan(1-x,z)\biggr):t,z\in\R,\,x<0\biggr\},
    \end{multline*}
    for $z=0$ and any fixed $t$, the limit
    point given by $x\to-\infty$ is represented by the boundary point
    $\left(\beta_2,W,\left(\frac{\pi}{2},
			-\frac{\pi}{2},\frac{\pi}{2},0\right)\right)$.
		Thus
		\[
		  \left[\left(\beta_2,W,\left(\frac{\pi}{2},
			-\frac{\pi}{2},\frac{\pi}{2},0\right)\right)\right]\in\sigma_2.
		\]
    Suppose that exists $p=[(\beta_1,W,(p',q',\theta,\phi)]\in\sigma_1$ so that 
		\[
		  p= \left[\left(\beta_2,W,\left(\frac{\pi}{2},
			-\frac{\pi}{2},\frac{\pi}{2},0\right)\right)\right].
		\]
    Then, by definition, it must be the case that
    every sequence $(x_j)\subset\dom(\beta_2)$ so that
    $\beta_2(x_j)\to \left(\frac{\pi}{2},\frac{-\pi}{2},\frac{\pi}{2},0\right)$
    is such that there exists a subsequence $(y_j)\subset (x_j)$
    so that $(y_j)\subset\dom(\beta_1)$ and $\beta_1(y_j)\to 
    (p',q',\theta,\phi)$.
    From above there exists a sequence $(x_i)\subset X$
    so that 
    $\beta_2(x_i)\to\left(\frac{\pi}{2},\frac{-\pi}{2},\frac{\pi}{2},0\right)$.
    By construction, however, $X\not\subset\dom(\beta_1)$ and
    therefore there is no suitable
    subsequence of $(x_i)$.
    Hence, for all $p\in\sigma_1$,
		$
		  p\neq \left[\left(\beta_2,W,\left(\frac{\pi}{2},
			-\frac{\pi}{2},\frac{\pi}{2},0\right)\right)\right].
		$
    Therefore $\sigma_1\neq\sigma_2$.

		Since $M=\dom(\beta_1)\cup\dom(\beta_2)$ and both
		$\ran(\beta_1)$ and $\ran(\beta_2)$ are compact, any sequence in 
		$M$ that does not converge to a point in $\sigma_i$ will converge to 
		a point in $\sigma_{j}$ where $i,j\in\{1,2\}$ with $i\neq j$. 
		Thus the union $\sigma=\sigma_1\cup\sigma_2$
		represents all of `the boundary' of $M$, as can be expected due to the 
		relationship with the conformal boundary. Hence
    $\sigma$ is complete.

    The extensions $(\beta_1,W)$ and $(\beta_2,W)$
		are compatible.
		Let $p=[(\beta_1,W,\{x\})]$ and $q=[(\beta_2,W,\{y\})]$.
		Assume that $p\contact q$, hence there exists
		$(z_i)\subset M$ so that $\beta_1(z_i)\to x$
		and $\beta_2(z_i)\to y$.
		Let $(x_i)\subset M$ be such that $\beta_1(x_i)\to x$.
		For each $i\in\N$, let $\hat{x}_{2i}=x_i$ and 
		$\hat{x}_{2i+1}=z_i$. By construction $\beta_1(\hat{x}_i)\to x$.
		The transition function, $\beta_2\circ\beta_1^{-1}$,
		amounts to the translation 
		and rotation used to transform between the two Cartesian 
		frames used to generate the spherical polar coordinates $\alpha_1$
		and $\alpha_2$. Therefore the transition function $\beta_2\circ\beta_1^{-1}$
		preserves distances. 
		This implies that
		$(\beta_2(\hat{x}_i))$ is Cauchy with respect to the Euclidean distance
		on $\ran(\beta_2)$. 
		Since $(\beta_2(\hat{x}_i))$
		is Cauchy, $(z_i)\subset(\hat{x}_i)$ and as
		$\beta_2(z_i)\to y$ it is the case that $\beta_2(\hat{x}_i)\to y$.
		As $(x_i)\subset(\hat{x}_i)$ the sequence $(\beta_2(x_i))$ converges to $y$.
		Therefore $q\covers p$. Repeating the argument with the roles of
		$x$ and $y$ interchanged demonstrates that $q=p$ as required.
		Hence $(\beta_1,W)$ and $(\beta_2,W)$
		are compatible.
		
		I compute a specific example. Consider the 
    line given by $L=\{(0,x,-x,1):x>0\}$. Thus
    \begin{multline*}
      \beta_1(L)=\Biggl\{\Biggl(\tan^{-1}\left(\sqrt{2x^2+1}\right),\,
				\tan^{-1}\left(-\sqrt{2x^2+1}\right),\, \\
				\cos^{-1}\left(\frac{1}{\sqrt{2x^2+1}}\right),\,
				-\frac{\pi}{4}\Biggr):x>0\Biggr\}
    \end{multline*}
		and 
    \begin{multline*}
      \beta_2(L)=\Biggl\{\Biggl(\tan^{-1}\left(\sqrt{2x^2+3}\right),\,
				\tan^{-1}\left(-\sqrt{2x^2+3}\right),\, \\
				\cos^{-1}\left(\frac{-x-1}{\sqrt{2x^2+3}}\right),\,
				\arctan\left(1-x,1\right)\Biggr):x>0\Biggr\}.
    \end{multline*}
    Hence, the boundary points representing the
    `end-point' of the line at infinity are
    $
      \left[\left(\beta_1,W,\left(\frac{\pi}{2},-\frac{\pi}{2},
        \frac{\pi}{2},-\frac{\pi}{4}\right)\right)\right]\in\sigma_1
    $
    and
    $
      \left[\left(\beta_2,W,\left(\frac{\pi}{2},-\frac{\pi}{2},
        \frac{3\pi}{4},\pi\right)\right)\right]\in\sigma_2.
    $
    Choosing a sequence $(x_i)\subset L$ with no limit point,
    $\beta_1(x_i)\to(\frac{\pi}{2},-\frac{\pi}{2},\frac{\pi}{2},-\frac{\pi}{4})$
    and $\beta_2(x_i)\to (\frac{\pi}{2},-\frac{\pi}{2},
      \frac{3\pi}{4},{\pi})$.
    The paragraph above implies that
    \[
      \left[\left(\beta_1,W,\left(\frac{\pi}{2},-\frac{\pi}{2},\frac{\pi}{2},-\frac{\pi}{4}\right)\right)\right]=
				\left[\left(\beta_2,W,\left(\frac{\pi}{2},-\frac{\pi}{2},\frac{3\pi}{4},\pi\right)\right)\right].
    \]

		Lastly, note that the points in $\sigma_2\setminus\sigma_1$ are
		exactly those points which represent the endpoints of $X$ and $Z$.
		These are the points lying on the `sphere at infinity' which cannot
		be expressed in the chart $\beta_1$. 

		I give an example of how these points can be calculated.
		Let $\gamma:(0,\infty)\to M$ be given, in $t,x,y,z$ coordinates,
		by $\gamma(\tau)=(\sqrt{(\tau+1)^2+1+\tau^2},-\tau,0,\tau)$. 
    This curve lies in $X$ for $\tau<0$.  
		From the calculation of $\beta_2(X)$,
		\begin{multline*}
		  \beta_2\circ\gamma(\tau)=
      \Biggl(\tan^{-1}\left(2\sqrt{(\tau+1)^2+1+\tau^2}\right),\,
				0,\,\\
				\cos^{-1}\left(\frac{-1}{\sqrt{(\tau+1)^2+1+\tau^2}}\right),\,
				\arctan(1+\tau,\tau)\Biggr)
		\end{multline*}
		so that as $\tau\to\infty$ the curve $\beta_2\circ\gamma(\tau)$ 
    limits to the boundary point
		\[
		  \left[\left(\beta_2,X,\left(\frac{\pi}{2},0,
        \frac{\pi}{2},\frac{\pi}{4}\right)\right)\right].
      \]
		The point $(\frac{\pi}{2},0,\frac{\pi}{2},\frac{\pi}{4})$
		is an element of $\mathcal{J}_2^+$. Since $\gamma\subset X$ and
		$X\not\subset\dom(\beta_1)$ this boundary point is an element of 
		future null infinity that cannot be represented by the chart $\beta_1$.
		
		Taking the set of all affinely parametrised geodesics as the
		b.p.p.\ satisfying set of curves, the 
		usual analysis of the Penrose conformal boundary implies that
		every point in $\sigma$ is a pure point at infinity.
  
  \subsection{Curzon solution}\label{sec:curzon}

    In this section I present the construction and classification of
    two representations of the boundary of
    the Curzon solution. The Curzon solution is one of 
    the better known and analysed examples of
    a directional singularity, \cite[Section 1]{Scott1986CurzonI}.
    I do not go into as 
    much detail as Section \ref{sec:Minkowski} and assume that, with 
    the results of Section \ref{sec:Minkowski} as a guide, the reader
    can fill in the details.
    
    The Curzon solution is the Weyl metric, \cite{Weyl1917Zur}, 
    for a monopole potential,
    \cite{Curzon1925Cylindrical}.
    The manifold is
    $
      M=\R\times\left(\R^2\setminus\{(0,0)\}\right)\times \Sph^1.
    $
    Letting $\phi^{-1}:(0,2\pi)\to \Sph^1$ be defined by 
    $\phi^{-1}(\theta)=(\cos\theta,\sin\theta)$ there is a chart,
    $\alpha$, given by
		$
		  \alpha(t,z,r,s)=(t,z,r,\phi(s)).
		$
		An additional chart is required to cover all of $M$ to account for the
		coordinate singularity introduced by $\phi$.
    In the chart $\alpha$ the metric takes the form,
    $
      ds^2=-\exp(2\lambda)dt^2+\exp(2\nu-2\lambda)\left(dr^2+dz^2\right)
        +r^2\exp(-2\lambda)d\phi^2,
    $
    where $\lambda=\frac{-m}{R}$, $\nu=\frac{-m^2r^2}{2R^4}$,
    $R=\sqrt{r^2+z^2}$ and $m\neq 0$, 
    \cite[Equations 1 and 2]{Scott1986CurzonI}. Note that in
    \cite[Equation 1]{Scott1986CurzonII} the metric has been 
    rescaled with respect to $m$
    so that
    $\lambda=\frac{-1}{R}, \nu=\frac{-r^2}{2m^2R^4}$ and 
    $R=\frac{\sqrt{r^2+z^2}}{m}$.
    This rescaling does not effect the coordinate transformations used 
    below nor the analysis of
    the global structure of the spacetime (compare the 
    metrics and transformations
    used in \cite{Scott1986CurzonI} and \cite{Scott1986CurzonII}).

    For the sake of concentrating on the structure of the
    boundary at $R=0$ I assume that $ds^2$ approaches the Minkowski
    metric for $R$ very large. As a consequence charts similar to those
    used in the previous section, Section \ref{sec:Minkowski}, 
    can be constructed to 
    induce a boundary $\sigma\subset\ABBPM{M}$ which is comprised
    entirely of pure points at infinity. This boundary
    will be such that if $(x_i)_i\subset M$ is a sequence
    of points so that if $R(x_i)\to\infty$ then there exists
    at least one element of $\sigma$ that is approached by $(x_i)$,
    whereas, if
    $R(x_i)\to 0$ then there is
    no element of $\sigma$ that is approached by $(x_i)$.
    Hence
    $\sigma$ is not complete.
    
    To find a complete boundary,
    boundary points that
    represent the `points at $R=0$' need to be constructed. 
    From the construction of $\alpha$,
    $
      \BP{\alpha}=\{(t,0,0,\phi):t\in\R,\phi\in(0,2\pi)\},
    $
    which are exactly the points for which $R=0$. Thus by taking
    $U=\R^3\times(0,2\pi)$ the set of boundary points,
    $   
      \sigma_\alpha=\{[(\alpha,U,\{(t,0,0,\phi)\})]:t\in\R,\phi\in(0,2\pi)\}
    $
    results. This gives a representation of the surface $R=0$.
    By construction, if $p\in\sigma_\alpha$ and $q\in\sigma$ then
    $p\parallel q$. Thus the extensions representing the set of boundary 
    points for 
    $r\to\pm\infty$, mentioned above,
    and $(\alpha,U)$ are compatible. From the definition of $M$, any sequence in
    $M$ that is without limit points must approach at least one 
    boundary point in $\sigma\cup\sigma_\alpha$ and therefore
    $\sigma\cup\sigma_\alpha$ is a complete boundary.
    
		Gautreau and Anderson, \cite{Gautreau1967Directional}, have shown that
		the Kretschmann  scalar, $K$, for the Curzon solution limits to $0$
		for curves on the $z$-axis
		as $R\to 0$ . But that it limits to $\infty$ for
		other straight line directions of approach to $R=0$. 
		Cooperstock and  Junevicus, 
		\cite{Cooperstock1974Singularities}, took this analysis further by 
		considering curves defined by
		$
		  \frac{z}{m}= C\left(\frac{r}{m}\right)^n
		$
		for $C,n>0$.
		They found that $\lim_{R\to 0}K\to 0$  for $0<n<\frac{2}{3}$ and that
		$\lim_{R\to0}K\to \infty$ for $n>\frac{2}{3}$.
		Scott and Szekeres, \cite[Section 3]{Scott1986CurzonI}, corrected a mistake
		that Cooperstock and Junevicus made for the critical, $n=\frac{2}{3}$, case 
		and showed that, when $n=\frac{2}{3}$, it was possible 
    to construct curves so that
		$\lim_{R\to 0}K$ takes any value in $\R^+\cup\{0,\infty\}$.
		This implies that the elements of $\sigma_\alpha$ are not regular points.		
		
		In order to take the classification of the points in 
    $\sigma\cup\sigma_\alpha$ further
    it is necessary to chose a b.p.p.\ satisfying 
    set of curves. Scott and Szekeres'
    analysis of the global structure of the Curzon solution, given in 
    \cite{Scott1986CurzonII},
    is based on    
    on timelike, null and spacelike geodesics. Hence I take these curves, with 
    affine parameters, as the set of b.p.p.\ satisfying curves.
    As seen above, however, the literature prior to 
    \cite{Scott1986CurzonII} used a variety of
    non-geodesic curves. While these curves cannot contribute 
    to the `approachability' of boundary points (because
    of the choice of the b.p.p.\ satisfying set of curves)
    the analysis of scalar quantities along them will still 
    give information about the regularity of the limit points of their
    images under the chart
    $\alpha$.
    
    Scott and Szekeres, \cite{Scott1986CurzonI}, show that there exist 
    spacelike geodesics that approach the $R=0$ surface
    with bounded affine parameter,
    \cite[Section 3.b]{Scott1986CurzonI}. Thus the elements 
    of $\sigma_\alpha$ are singular points.
		It would be nice to conclude from the analysis of the Kretschmann scalar
		that the elements of $\sigma_\alpha$
		are directional singularities. While the results of the
		analysis are suggestive of this, they do not
		prove it. In order to show that
		$[(\alpha,U,\{p\})]\in\sigma_\alpha$ is a directional 
		singularity it is
		necessary to show two things.
		First, that $(\alpha,U,\{p\})$ is not covered by non-singular points.
		Second, a regular boundary point $(\beta,X,\{q\})$
		so that $(\alpha,U,\{p\})\covers(\beta,X,\{q\})$ needs to be constructed. 
    The analysis of the
		Kretschmann scalar proves the first but not the second condition.
		Fortunately, Scott and
		Szekeres have constructed a chart (presented as $\beta$ below) which 
		produces a regular point covered by every element of $\sigma_\alpha$. 
		Thus the elements of $\sigma_\alpha$ are indeed directional 
		singularities. 

    For the purposes of studying the directional behaviour
    Scott and Szekeres restrict to the half plane $z\geq 0$ (as the metric
    has a symmetry under $z\mapsto -z$) and
    take two coordinate transformations.
    The first is, \cite[Equations 17 and 18]{Scott1986CurzonI},
    \begin{align*}
      x(r,z) &= \tan^{-1}\left(\frac{r}{m}\exp\left(\frac{m}{z}\right)\right)
        +\tan^{-1}\left(\frac{r}{m}
          \exp\left(-\left(\frac{\sqrt{2} m}{r}\right)^\frac{2}{3}\right)\right)
          ,\\
      y(r,z) &= \tan^{-1}\left(3\frac{z}{m}-\left(\frac{z}{m}\right)^2
				\frac{R \exp\left(\nu-\lambda\right)}
        {\left(R^{12}+R^4+\frac{1}{3}\left(\frac{r}{m}\right)^2\right)^{\frac{1}{4}}}\right),
    \end{align*}
    where $(x,y)\in
    (-\frac{1}{2}\pi,\frac{1}{2}\pi)\times(-\frac{1}{2}\pi,0]\cup(-\pi,\pi)\times(0,\frac{\pi}{2})$.
    These coordinates were constructed using a combination of 
    numerical calculation of geodesics, physical intuition and trail and 
    error \cite[page 566]{Scott1986CurzonI}.
    The second transformation is, \cite[Equations 7 and 8]{Scott1986CurzonII}, 
    \begin{align*}
      Y(t,r,z) &= \frac{\pi}{2}+\tan^{-1}
				\left(ay^3\left(x^2-\frac{\pi^2}{4}\right)^2\frac{m}{z}+
				3\frac{z}{m}-\vphantom{\frac{R\exp(v-\lambda)}
        {\left(R^{12}\left(1+\left(\frac{t}{m}\right)^4\right)+R^4+\frac{1}{3}
          \left(\frac{r}{m}\right)^2\right)^{\frac{1}{4}}}}\right.\\
        &\hspace{3cm}\left.\left(\frac{z}{m}\right)^2\frac{R\exp(v-\lambda)}
        {\left(R^{12}\left(1+\left(\frac{t}{m}\right)^4\right)+R^4+\frac{1}{3}
          \left(\frac{r}{m}\right)^2\right)^{\frac{1}{4}}}\right),\\
      T(t,r,z) &= \tan^{-1}\left(\exp(-\hat K)\left(\frac{t}{m}+H\right)+
        \frac{t}{m}\left(y+\frac{\pi}{2}\right)^3\right)+\\
          &\hspace{3cm}\tan^{-1}\left(\exp(-\hat K)\left(\frac{t}{m}-H\right)+
        \frac{t}{m}\left(y+\frac{\pi}{2}\right)^3\right),
    \end{align*}
    where
    \[
      H(r,z)=\frac{1}{2}\left(\frac{r}{z}\right)^2\exp\left(\frac{2m}{z}\right)
      +
      \int_1^{\frac{z}{m}}\exp\left(\frac{2}{u}\right)du
    \]
    and
    \[
      \hat K(r,z)=
        \left(y+\frac{\pi}{2}\right)R+\left(\frac{\tan x}{\tan y}\right)^2.
    \]
    Note that in \cite{Scott1986CurzonI,Scott1986CurzonII} the function 
    $\hat K$ is denoted as $K$. I use $\hat K$ here to distinguish 
    it from the Kretschmann scalar. Please refer to
    \cite[Section 3]{Scott1986CurzonII} for a discussion of the motivation
    behind these coordinates.
    
    The coordinate ranges are
    $T\in (-\pi,\pi)$ and 
    $(x,Y)\in \left(-\frac{\pi}{2},\frac{\pi}{2}\right)
    \times\left(0,\frac{\pi}{2}\right]\cup
    (-\pi,\pi)\times
    \left(\frac{\pi}{2},\pi\right)$.
    This defines a chart
    \begin{equation*}
      \beta(t,z,r,s)=(T(t,r,z),\,x(r,z),\,Y(t,r,z),\,\phi(s))
    \end{equation*}
    so that
    $\dom(\beta)=\dom(\alpha)$ and 
    \begin{equation}
    	\ran(\beta)=(-\pi,\pi)\times
    	\Biggl(
    	\left(-\frac{\pi}{2},\frac{\pi}{2}\right)\times\left(0,\frac{\pi}{2}\right]
    	\cup
    	(-\pi,\pi)\times\left(\frac{\pi}{2},\pi\right)
    	\Biggr)\times(0,2\pi).\label{eq:beta}
    \end{equation}
		Figure \ref{fig:4} shows the range of $\beta$ with the $\phi$ coordinate 
		suppressed.
		
		The construction of $\beta$ is based on Scott and Szekeres' detailed
		analysis of the behaviour the Kretschmann scalar along 
		certain sets of curves, \cite{Scott1986CurzonI,Scott1986CurzonII}.
		Hence their construction of this chart is an example of the technique
		used to study the global structure of a manifold mentioned in the 
		introduction to this paper.
		
		\begin{figure*}
		  \centering
		  \def\svgwidth{0.6\columnwidth}
		  \input{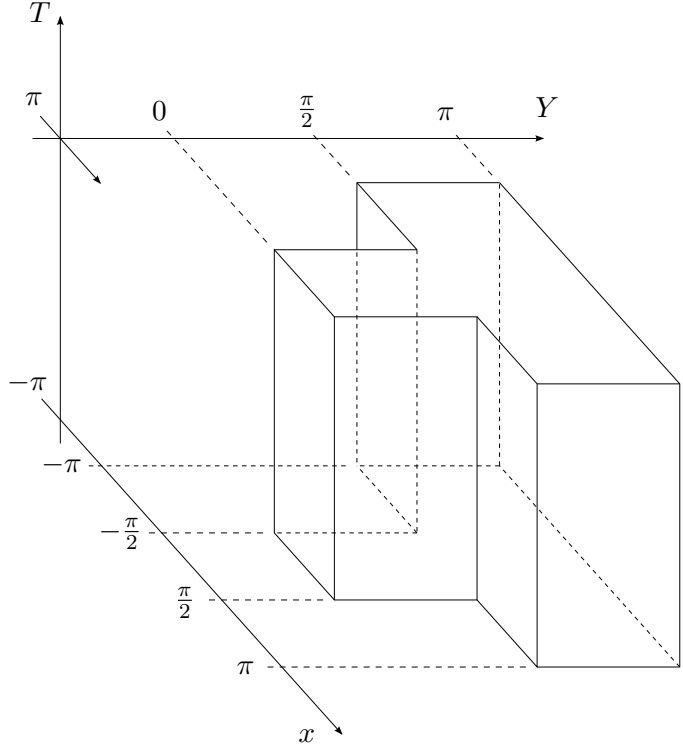}
		  \caption{The range of $\beta$ (Equation \eqref{eq:beta}), with 
        the $\phi$ coordinate suppressed, 
		    is the interior of the T shaped set.}\label{fig:4}
	  \end{figure*}

    I will now discuss the set of boundary points induced
    by an extension of $\beta$. Because it reduces the
    complexity of this discussion I will consider $\beta$ as a chart on
    the $z>0$ submanifold rather than as a chart on $M$ itself.
    The topological boundary of the range of $\beta$ can be divided
    into the eleven surfaces
    \begin{align*}
      S_1&=\Bigl\{\left(T,x,\pi,\phi\right):T,x\in\left(-\pi,\pi\right),\,\phi\in\left(0,2\pi\right)\Bigr\},\\
      S_2^+&=\left\{\left(T,\pi,Y,\phi\right):T\in\left(-\pi,\pi\right),\,Y\in\left(\frac{\pi}{2},\pi\right),\,\phi\in\left(0,2\pi\right)\right\},\\
      S_2^-&=\left\{\left(T,-\pi,Y,\phi\right):T\in\left(-\pi,\pi\right),\,Y\in\left(\frac{\pi}{2},\pi\right),\,\phi\in\left(0,2\pi\right)\right\},\\
      S_3^+&=\left\{\left(T,x,\frac{\pi}{2},\phi\right):T\in\left(-\pi,\pi\right),\,x\in\left(\frac{\pi}{2},\pi\right),\,\phi\in\left(0,2\pi\right)\right\},\\
      S_3^-&=\left\{\left(T,x,\frac{\pi}{2},\phi\right):T\in\left(-\pi,\pi\right),\,x\in\left(-\pi,-\frac{\pi}{2}\right),\,\phi\in\left(0,2\pi\right)\right\},\\
      S_4^+&=\left\{\left(T,\frac{\pi}{2},Y,\phi\right):T\in\left(-\pi,\pi\right),\,Y\in\left(0,\frac{\pi}{2}\right),\phi\in\left(0,2\pi\right)\right\},\\
      S_4^-&=\left\{\left(T,-\frac{\pi}{2},Y,\phi\right):T\in\left(-\pi,\pi\right),\,Y\in\left(0,\frac{\pi}{2}\right),\phi\in\left(0,2\pi\right)\right\},\\
      S_5^-&=\left\{\left(T,x,0,\phi\right):T\in\left(-\pi,0\right),x\in\left(-\frac{\pi}{2},\frac{\pi}{2}\right),\,\phi\in\left(0,2\pi\right)\right\},\\
      S_5^+&=\left\{\left(T,x,0,\phi\right):T\in\left(0,\pi\right),x\in\left(-\frac{\pi}{2},\frac{\pi}{2}\right),\,\phi\in\left(0,2\pi\right)\right\},\\
      S_6&=\biggl\{\left(\pi,x,Y,\phi\right):\phi\in\left(0,2\pi\right),\\
         &\hspace{1cm}\left(x,Y\right)\in\left(-\pi,\pi\right)\times\left(\frac{\pi}{2},\pi\right)\cup\left(-\frac{\pi}{2},\frac{\pi}{2}\right)\times\left(0,\frac{\pi}{2}\right]\biggr\},\\
		  S_{7}&=\biggl\{\left(-\pi,x,Y,\phi\right):\phi\in\left(0,2\pi\right),\\
        &\hspace{1cm}\left(x,Y\right)\in\left(-\pi,\pi\right)\times\left(\frac{\pi}{2},\pi\right)\cup\left(-\frac{\pi}{2},\frac{\pi}{2}\right)\times\left(0,\frac{\pi}{2}\right]\biggr\},
    \end{align*}
    and the lines 
    or points given by the intersections of the closures of the surfaces.
    See Figure \ref{fig:5} for a graphical representation.
    All of these surfaces and the points or lines
    given by intersections of the closures of the surface, consist
    of admissible boundary points. Let $X=\R^4$ then $(\beta, X)$
    is an extension. 
    Considering $\beta$ as a chart in $M$, rather than on the submanifold
    given by $z>0$, the elements of $S^\pm_3$ are no longer admissible. See 
    \cite[Section 5]{Scott1986CurzonI} or the third and 
    second to last sentences on
    page 578 of \cite{Scott1986CurzonII}.
    
    \begin{figure*}
		  \centering
		  \def\svgwidth{0.7\columnwidth}
		  \input{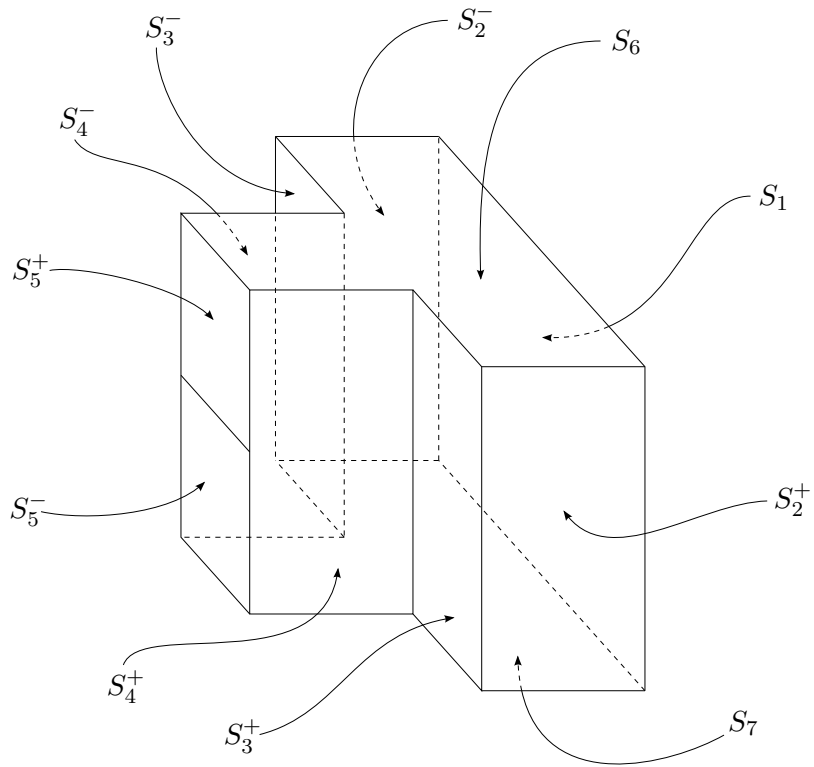}
		  \caption{A graphical representation of the eleven surfaces that make up
		  the topological boundary of the range of $\beta$, see 
      Equation \eqref{eq:beta}.}\label{fig:5}
	  \end{figure*}
    
    This second chart, $\beta$, also induces a boundary, $\sigma_\beta$,
    as the set of boundary points
    $[(\beta,X,\{p\})]$ where $p$ is an element of one of 
    the eleven surfaces or an element of one 
    of the intersections of the closures of 
    the surfaces given above. 
    
    Scott and Szekeres do not explicitly state all 
    the information needed for a complete classification 
    of the boundary points of $\sigma_\beta$
    in \cite{Scott1986CurzonI,Scott1986CurzonII}. 
    The following list contains those statements
    that can be derived from comments in their papers.
    Note that the analysis
    of the two papers is performed in different coordinate systems,
    $x,y,t,\phi$ versus $x, Y, T, \phi$. This difference, however, does not
    effect the conclusions drawn below (for justification refer to
    the second sentence on page 579 of \cite{Scott1986CurzonII});
    \begin{enumerate}
      \item the set $S_1\cup S_2^+\cup S_2^-$ and the lines 
				$\overline{S_1}\cap\overline{S_2^\pm}\setminus\overline{S_6\cup S_7}$ 
        correspond to 
				spacelike infinity for the 
        $(t,r,z,\phi)$ coordinates, i.e. as $R\to\infty$. That is, these
        surfaces are approached by spacelike geodesics and not
        approached by non-spacelike geodesics.
        See
        the sentence beginning ten lines from the top of page 575 of
        \cite{Scott1986CurzonII},
      \item the surfaces $S_3^\pm$ consist of $C^\infty$ regular points.
        See 
        \cite[Section 5]{Scott1986CurzonI} or 
        the third and second to last sentences on
        page 578 of \cite{Scott1986CurzonII}.
      \item the lines given by $\overline{S_3^\pm}\cap \overline{S_4^\pm}$
				are the singularity of the Curzon solution. The Kretschmann 
        limits to $\infty$ along every curve that limits to points
        in $\overline{S_3^\pm}\cap \overline{S_4^\pm}$.
        This is stated in several places
        in \cite{Scott1986CurzonII} the clearest is the last sentence
        of page 571,
      \item the set $S_4^+\cup S_4^-$ is not approached by any 
        spacelike geodesic. See
        item 3 on page 568 of \cite{Scott1986CurzonI} and
        the first paragraph of \cite[Section 1]{Scott1986CurzonII},
      \item the lines 
        $\left(\overline{S_4^\pm}\cap
          \overline{S_5^+\cup S_5^-}\right)\setminus\overline{S_6\cup S_7}$ 
        and $\overline{S_5^+}\cap\overline{S^-_5}$ are
        spacelike infinities, i.e.\ they are
        approach by spacelike geodesics and
        not approached by non-spacelike geodesics. See items 2 and 4 on page 568 of
        \cite{Scott1986CurzonI},
      \item the sets $S_5^+, S_5^-$
        are comprised of $C^\infty$ regular points. 
        On these surfaces copies of Minkowski
        spacetime can be smoothly attached. See the last paragraph on page 579
        of \cite{Scott1986CurzonII}.
    \end{enumerate}

    The following statements are inferred from
    \cite{Scott1986CurzonI,Scott1986CurzonII}
    but are not explicitly mentioned in either of the papers.
    Since $S_3^\pm$ consists of regular boundary points the
    elements of $\ABBPM{M}$ represented by points in
    $\left(\overline{S_3^\pm}\cap\overline{S_2^\pm}\right)
      \setminus\overline{S_6\cup S_7}$
    are implied to be in $\sigma$. 
    Nothing is said about future/past timelike/null infinity, though
    Figure 1 of \cite{Scott1986CurzonII} implies that these will correspond to
    $S_6$ and $S_{7}$ and the edges $\overline{S_6}\setminus S_6$ and
    $\overline{S_7}\setminus S_7$. Lastly, nothing is explicitly said about 
    the existence of non-spacelike
    geodesics that approach the $S_4^\pm$ surfaces. In principle this
    information can be determined from \cite[Section 2]{Scott1986CurzonII}
    and the equations for the coordinates given above. It is, however,
    implied that no non-spacelike geodesic approaches
    the $S_4^\pm$ surfaces in the discussion of
    \cite[Section 2]{Scott1986CurzonII} beginning in the last paragraph
    of page 573. Thus I shall assume that the elements of
    $S_4^\pm$ are non-approachable.
    These claims could be investigated
    thoroughly. This is beyond the scope of this paper,  however, so I
    will assume that the analysis, as described above, is accurate.
    
    With this in mind then, $[(\beta,X,\{p\})]\in\sigma_\beta$,
    is classified as;
    \begin{enumerate}
			\item an unapproachable point if $p\in S_4$,
			\item a $C^\infty$ indeterminate point if $p\in S_5^\pm$,
      \item a pure point at infinity if $p\in\overline{S_1},\overline{S_2^\pm},
        \overline{S_3^\pm}\cap\overline{S_2^\pm},
				\overline{S_6},\overline{S_{7}}, \overline{S_5^+}\cap\overline{S_5^-},
				\overline{S_4^\pm}\cap\left(\overline{S_5^+\cup S_5^-}\right)$,
			\item a pure singularity if 
        $p\in\overline{S_3^\pm}\cap \overline{S_4^\pm}$.
    \end{enumerate}
    Thus the material of \cite{Scott1986CurzonI,Scott1986CurzonII} 
    is sufficient to produce a complete
    boundary, $\sigma_\beta$, and, with some inferences,
    its classification.
    The boundary $\sigma_\beta$ can be considered to be a
    `good' boundary since none of the global structure is unresolved, i.e.\
    the boundary contains no mixed points.

		I now return to the first boundary, $\sigma\cup\sigma_\alpha$.
		By construction each boundary point of $\alpha$ must cover the surfaces
		$S_3^\pm,S_4^\pm,S_5^\pm$ and the intersections of their closures 
		(minus the closures of $S_6$ and $S_7$). Indeed, Figure 2 of 
    \cite{Scott1986CurzonII}
		implies that the function $T$ extends to a bijective function
		on the boundary of the range of 
		$\beta$ 
    (see also the last paragraph of page 571 of \cite{Scott1986CurzonII}).
		Hence, for all $[(\alpha,U,\{(t,0,0,\phi)\})]\in\sigma_\alpha$, there
		is
		a point $(T,x,Y,\phi)$ lying in any one of the surfaces 
    $S_3^\pm,S_4^\pm,S_5^\pm$,
		or the intersections of their closures 
    (minus the closures of $S_6$ and $S_7$),
		such that
		$
		  [(\alpha,U,\{(t,0,0,\phi)\})]\covers [(\beta,U,\{(T,x,Y,\phi)\})].
		$
		In particular each boundary point in $\sigma_\alpha$ covers a regular point 
		lying in $S_5^\pm$. This justifies the claim that 
    $\sigma_\alpha$ is composed of directional singularities.
    
  \subsection{Smooth Gowdy symmetric generalized 
    Taub-NUT spacetimes}\label{ssec.Gowdy}
  
		The last example comes from Beyer and Hennig's work on the existence,
		and global properties, of smooth Gowdy symmetric generalized Taub-NUT 
		spacetimes, \cite{Beyer2011Smooth}. A smooth Gowdy symmetric generalized
		Taub-NUT spacetime is a generalization of a spacetime 
    in Moncrief's class of generalized
		Taub-NUT solutions of the vacuum Einstein equations on 
    $(0,\pi)\times \Sph^3$
		with Gowdy, 
		$U(1)\times U(1)$, symmetry such that the `surface' at 
    $\{0\}\times\Sph^3$ is a 
		smooth past Cauchy horizon, \cite[Section 3.2]{Beyer2011Smooth}.
		
		Beyer and Hennig's analysis is based on the study of a set of differential 
		equations. They do not solve these equations to produce a closed form
		equation for the metric.   
    Nevertheless a great deal can still be deduced
		about the chart induced boundaries of this class of spacetimes. 
    Due to the lack of a
		closed form of the metric additional work would be required
		before other boundary constructions could be applied.
		
		Because of the symmetries of the 
		spacetimes, and because of the choices made with regards to 
		the killing vectors used \cite[Section 2.2]{Beyer2011Smooth}, 
		the following discussion can be reduced to the coordinates $t\in(0,\pi)$
		and $\theta\in(0,\pi)$ where $\theta$ is the coordinate on 
		$\Sph^3=\{(x_1,x_2,x_3,x_4)\in\R^4:x_1^2+x_2^2+x_3^2+x_4^2=1\}$
		implicitly determined by, \cite[Section 2.2]{Beyer2011Smooth},
		\begin{align*}
		  x_1&=\cos\frac{\theta}{2}\cos\lambda_1, 
         & x_2&=\cos\frac{\theta}{2}\sin\lambda_1,\\
		  x_3&=\sin\frac{\theta}{2}\cos\lambda_2, 
         & x_4&=\sin\frac{\theta}{2}\sin\lambda_2.
		\end{align*}
		Let $\alpha$ be the chart defined by these coordinates.
		Let $U=\R\times(0,\pi)$
		so that $\BP{\alpha}\cap U=\BP{\alpha}=\{0,\pi\}\times(0,\pi)$
		and $(\alpha,U)\in\EX{M}$.		
		
		Beyer and Hennig demonstrate the global existence of solutions
    to the differential equations
		by first proving local existence for the Fuchsian system, 
    \cite[Equations 35 and 36]{Beyer2011Smooth},
		\begin{align*}
		  D^2 S-t^2\Delta_{\Sph^2}S &= (1-t\cot t)DS-\exp(-2S)
      \left(\left(D\omega\right)^2-
        \left(t\partial_\theta\omega\right)^2\right),\\
		  D^2\omega - 4D\omega - t^2\Delta_{\Sph^2}\omega &= 
        (1-t\cot t)D\omega +2(D S-2)D\omega -2
        (t\partial_\theta S)(t\partial_\theta\omega),
		\end{align*}
		where $D=t\partial_t$ and
		$
		  \Delta_{\Sph^2}=
        \partial_\theta^2+\cot\theta\partial_\theta+
        \frac{1}{\sin^2\theta}\partial_\phi^2
		$
		is the Laplace operator on the unit sphere.
		The functions $S,\omega$ are related to the metric components,
		\cite[Section 2 and 3]{Beyer2011Smooth}.
		Global existence is then implied by 
    \cite[Theorem 6.3]{Chrusciel1990Spacetimes}.
		
		In Section 3 of \cite{Beyer2011Smooth}, Beyer and
		Hennig show that the solution to these equations can be extended to 
    $(-\epsilon,\pi)\times\Sph^3$ 
		for some $\epsilon>0$. Note, however, that on $(-\epsilon,0)\times\Sph^3$
		the extended solution is no longer guaranteed to be a solution of the 
		vacuum Einstein equations. In any case,
		the result is that the surface given by
		$t=0$ is composed of $C^\infty$ regular boundary points. That is
		the boundary points in
		$
		  \{[(\alpha,U,\{(0,\theta)\})]:\theta\in(0,\pi)\}
		$
		are indeterminate points, Definition \ref{prop.App&UnAppABBPM}.
    Since the elements $(0,0)$ and $(0,\pi)$ are 
		not boundary
		points of $\alpha$ it is not possible to conclude that they
		are also regular boundary points. 
    This problem has been caused by the coordinate singularity
    in the definition of $\alpha$.
		Beyer and Hennig's analysis does not suffer this restriction,
		hence it is clear that if the $\theta$ axis was rotated to define
		a new chart the points corresponding to $(0,0)$ and $(0,\pi)$
		would be $C^\infty$ regular.
		
		To study the global behaviour of solution's to the Fuchsian system
		Beyer and Hennig recast it using an Ernst potential, $\mathcal{E}$, that
		solves the equation, \cite[Equation 58]{Beyer2011Smooth},
		\[
		  \left(-\partial_t\mathcal{E}-\cot t\partial_t\mathcal{E}+\partial_\theta^2\mathcal{E}+\cot\theta\partial_\theta\mathcal{E}\right)f=-(\partial_t\mathcal{E})^2+(\partial_\theta\mathcal{E})^2,
		\]
		where $f$ is the real part of $\mathcal{E}$.
		This is equivalent to a linear partial differential system
		which reduces to a linear ordinary differential system on
		each of the surfaces $t=0$, $\theta=0$ and $\theta=\pi$, 
		\cite[Equation 75]{Beyer2011Smooth}. By solving this linear ODE 
		Beyer and Hennig are able to study the properties of the Ernst potential
		on the surface $t=\pi$. It turns out that the
		behaviour of the Ernst potential on the surface $t=\pi$ depends on
		two parameters $b_A$ and $b_B$ which are related to the initial data
		of the Ernst equation, \cite[Section 4.3.2]{Beyer2011Smooth}.
		Beyer and Hennig make conclusions regarding the
		$t=\pi$ surface by dividing the behaviour of the Ernst potential into four 
		cases, \cite[Section 4.4]{Beyer2011Smooth};
		\begin{description}
		  \item[$b_A=b_B$:] The surface $t=\pi$ is a 
        $C^\infty$ regular Cauchy horizon. 
				Thus the boundary points $(\alpha,U,\{(\pi,\theta)\})$, 
        $\theta\in(0,\pi)$, are $C^\infty$
				regular boundary points, \cite[Section 4.4.1]{Beyer2011Smooth}.
				As before if a new chart is introduced, by rotating the $\theta$ axis,
				the boundary points corresponding to $(\pi,0)$ and $(\pi,\pi)$
				would be $C^\infty$ regular.
				
			\item[$b_A\neq b_B$ and $b_B\neq b_A\pm 4$:] In this case one of the
				metric components diverges on the surface $t=\pi$. The 
        Ernst potential is, however,
				regular everywhere and the Kretschmann
				scalar is bounded on $t=\pi$. 
				To analyse this further Beyer and Hennig construct a new chart, which 
  			will be denoted by $\beta$. With respect to this new chart 
				the surface given by $t=\pi$ is a regular Cauchy horizon, see
				Equation (117), and the comments immediately before Equation 
				(117), of
				\cite{Beyer2011Smooth}. 
				Each boundary point $(\alpha,U,\{(\pi,\theta)\})$
				is therefore covered by the set of $C^\infty$ regular boundary points 
				$
				  \{(\beta,X,(0,\theta):\theta\in(0,\pi)\},
				$
				where $(\beta,X)$ is a suitable extension. The boundary
				points $[(\alpha,U,\{(\pi,\theta)\})]$ are therefore indeterminate
				boundary points, Definition \ref{prop.App&UnAppABBPM}.
				
				Again, Beyer and Hennig's analysis indicates that
				if the $\theta$ axis were rotated and the same analysis performed
				the points corresponding
				to $(\pi,0)$ and $(\pi,\pi)$ would be removable.
				
			\item[$b_B=b_A+4$:] In this case the Ernst potential diverges on 
				the surface $\theta=0$, \cite[Section 4.4.2]{Beyer2011Smooth}. 
				The Kretschmann scalar on 
				the surface $\theta=0$ behaves like $\frac{1}{(\pi-t)^{12}}$
				as $t\to\pi$, \cite[Section 4.4.2]{Beyer2011Smooth}.
				Because of this Beyer and Hennig are unable to use the 
				techniques used in the
				previous two cases. To avoid this issue they take a sequence of
				solutions with $b_B\neq b_A+4$ that converge to the
				$b_B=b_A+4$ case. The result is that the Ernst potential that is the
				limit of the Ernst potentials of the sequence of solutions
				is regular for $0<\theta\leq\pi$ and diverges for $\theta=0$,
				\cite[Section 4.4.2]{Beyer2011Smooth}. The 
				implication being that the boundary points
				$
				  \{(\alpha,U,\{(\pi,\theta)\}):\theta\in(0,\pi)\}
				$
				are $C^\infty$ regular boundary points. 
				
				Their analysis also shows that, if the 
				$\theta$ axis were rotated the point corresponding
				to $(\alpha,U,\{(\pi,\pi)\})$ would be a $C^\infty$ regular point
				and that the point corresponding to $(\alpha,U,\{(\pi,0)\})$ would
				be an essential singularity. Beyer and Hennig's analysis in 
				\cite{Beyer2011Smooth}
				does
				not show if $[(\alpha,U,\{(\pi,0)\})]$ is pure or directional.
								
		  \item[$b_B=b_A-4$:] The analysis of this case is exactly the same as
				for $b_B=b_A+4$ except that the divergence occurs in the limit to 
				the point
				$(\pi,\pi)$ along the surface $\theta=\pi$.
		\end{description}

\appendix

\section{Properties of the completion of a manifold with respect to sets
of extensions}\label{sec.topQ(M)}

  In Section \ref{sec.cicib} it was claimed that;
  \begin{enumerate}
    \item the quotient map $q$ is open, continuous 
      and such that its restriction, for all $(\alpha, U)\in S_Q$,
      to $N_{(\alpha, U)}$ is a homeomorphism,
    \item the space $Q(M)$ is a $T_1$, separable, first countable, 
      locally metrizable
      topological space, and, 
    \item that there exists a continuous map
       $\imath_Q:M\to Q(M)$ that is a homeomorphism onto its image so that
       $\overline{\imath_Q(M)}=Q(M)$.
  \end{enumerate}
  This section presents proofs of these claims. Throughout
  this section I assume that $Q\subset \EX{M}$ and that
  $N_Q, S_Q, Q(M)$ and $q:N_Q\to Q(M)$ are as given in Definition
    \ref{def.quotientmap}.

  \begin{proposition}
    The map
    $q$ is open.
  \end{proposition}
  \begin{proof}
    It is sufficient to show that $q^{-1}(q(W))$ is open
    for any open $W\subset N_{(\alpha,U)}$. Suppose
    otherwise. 
    Then there exists $(\beta,X)\in S_Q$
    so that $N_{(\beta, X)}\cap q^{-1}(q(W))$ is not open.
    Since $N_{(\beta, X)}\cap q^{-1}(q(W))$ is first countable
    there exists
    $x\in q^{-1}(q(W))\cap N_{(\beta,X)}$
    and $(x_i)\subset 
    N_{(\beta, X)}\setminus q^{-1}(q(W))$
    so that $x_i\to x$ in $N_{(\beta, X)}$.
    Since $q(x)\in q(W)$ there exists $y\in W$ so that $q(x)=q(y)$.
    By the definition of $Q(M)$ this implies that
    there exists a subsequence $(y_i)$ of $(x_i)$ so that
    $\alpha\circ\beta^{-1}(y_i)\to y$. Since $W$ is open
    there exists some $i$ so that $\alpha\circ\beta^{-1}(y_i)\in W$.
    By construction there exists some $x_j$ so that $x_j=y_i$.
    Thus $q(x_j)=q(\alpha\circ\beta^{-1}(y_i))\in q(W)$.
    This is a contradiction and therefore $q$ is an open map.
  \end{proof}
 
  \begin{corollary}
    The space $Q(M)$ is $T_1$.
  \end{corollary}
  \begin{proof}
    Let $[x],[y]\in Q(M)$. Then there exists $(\alpha, U)\in S_Q$
    so that
    $x\in N_{(\alpha, U)}$. If $[y]\in q(N_{(\alpha, U)})$
    then there exists $z\in N_{(\alpha, U)}$ so that
    $q(z)=q(y)$. Since $N_{(\alpha, U)}$ is Hausdorff and
    as $q$ is open it is clear that either $[x]=[y]$ or $[x]$ and $[y]$ are
    $T_1$ separated. So suppose that
    $[y]\not\in q(N_{(\alpha, U)})$. 
    There exists $(\beta, X)\in S_Q$ so that
    $y\in N_{(\beta, X)}$.
    By symmetry
    we can assume that $[x]\not\in q(N_{(\beta, X)})$.
    By construction $N_{(\alpha, U)}$ and $N_{(\beta, X)}$
    are open subsets of $N_Q$. Since $q$ is open
    we have the required open sets and $Q(M)$ is $T_1$.
  \end{proof}

  \begin{corollary}\label{cor.firstcoutnable}
    The space $Q(M)$ is first countable.
  \end{corollary}
  \begin{proof}
    The space $N_Q$ is clearly first countable. Since
    $Q(M)$ is the image of an open continuous map
    it is also first countable, \cite[Problem 16A.3]{citeulike:593505}.
  \end{proof}

  \begin{proposition}\label{prop.qrestishomeo}
    For all $(\alpha, U)\in S_Q$ the map
    $q|_{N_{(\alpha, U)}}$ is a homeomorphism.
  \end{proposition}
  \begin{proof}
    The map $q$ is continuous by definition of the topology on
    $Q(M)$.
    From above $q$ is open. Surjectivity of $q|_{N_{(\alpha, U)}}$
    follows as the image of this map is
    $q(N_{(\alpha, U)})$. Injectivity follows from the construction
    of $Q(M)$ and as each $N_{(\alpha, U)}$ is Hausdorff.
  \end{proof}

  \begin{corollary}
    The space $Q(M)$ is locally metrizable.
  \end{corollary}
  \begin{proof}
    Let
    $d_{(\alpha, U)}:N_{(\alpha, U)}\times N_{(\alpha, U)}\to\R$
    be the distance on $N_{(\alpha, U)}$ induced 
    by the euclidean distance on $\R^n$
    and the inclusion $N_{(\alpha, U)}\subset\R^n$.
    Since $N_{(\alpha, U)}$ was given the relative topology with
    respect to this inclusion the distance $d$ is compatible
    with the topology on $N_{(\alpha, U)}$.
    Define $d:q(N_{(\alpha, U)})\times q(N_{(\alpha, U)})\to\R$ by
    \[
      d([x], [y]) = d_{(\alpha, U)}\left(
        \left(q|_{N_{(\alpha, U)}}\right)^{-1}(x),
        \left(q|_{N_{(\alpha, U)}}\right)^{-1}(y)\right).
    \]
    It can easily be checked that $d$ is a distance on 
    $q(N_{(\alpha, U)})$. Since
    $q$ is open $d$ is compatible with the
    topology on $Q(M)$.
    Since $Q(M)$ is covered by 
    the set $\{q(N_{(\alpha, U)}):(\alpha, U)\in S_Q\}$, 
    $Q(M)$ is locally metrizable.
  \end{proof}

  \begin{proposition}\label{prop.imathconstruct}
    There exists an injective continuous function $\imath_Q:M\to Q(M)$
    that is a homeomorphism onto its image. The image of $M$ under
    $\imath_{Q}$ is an open dense subset of $Q(M)$.
  \end{proposition}
  \begin{proof}
    For each $x\in M$ choose some pair
    $(\alpha_x, U_x)\in S_Q$ so that
    $x\in\dom(\alpha_x)$, then $\alpha_x(x)\in N_Q$.
    Define $\imath_Q(x)=q(\alpha_x(x))$.
    The definition of $Q(M)$ implies that $\imath_Q$ is well defined
    and independent of the choice of $\alpha_x$. In particular, due to the
    definition of $Q(M)$,
    if $x\in M$ then $\imath_Q(x)=[\alpha(x)]$ 
    for any $\alpha\in\Atlas{M}$
    so that $x\in\dom(\alpha)$.

    Suppose that $x,y\in M$ are such that $\imath_Q(x)=\imath_Q(y)$.
    That is $[\alpha_x(x)]=[\alpha_y(y)]$, by definition this implies
    that $\alpha_y\circ\alpha_x^{-1}(\alpha_x(x))=\alpha_y(y)$.
    Thus $x=y$. Hence $\imath_Q$ is injective.

    I now show that $\imath_Q$ is continuous. 
    Let $V\subset Q(M)$ be open. By definition
    $q^{-1}(V)$ is open, hence for each 
    pair $(\alpha,U)\in S_Q$ the set
    $\alpha^{-1}(N_{(\alpha, U)}\cap q^{-1}(V))\subset M$ is open.
    I claim that
    \[
      \imath_Q^{-1}(V)=\bigcup_{{(\alpha,U)}\in S_Q}
      \alpha^{-1}(N_{(\alpha,U)}\cap q^{-1}(V)).
    \]
    If  this is true then $\imath_Q^{-1}(V)$ is open
    and therefore $\imath_Q$ is continuous.

    I now prove the claim. Let $x\in \imath_{Q}^{-1}(V)$,
    then $x\in M$ hence there exists $(\alpha, U)\in S_Q$ so that
    $x\in\dom(\alpha)$. Therefore $\imath_{Q}(x)=[\alpha(x)]\in V$.
    This implies that $\alpha(x)\in q^{-1}(V)$ and hence
    $x=\alpha^{-1}(\alpha(x))\in 
      \alpha^{-1}(N_{(\alpha, U)}\cap q^{-1}(V))$.
    That is $\imath_{Q}^{-1}(V)\subset
    \bigcup_{(\alpha, U)\in S_Q}
    \alpha^{-1}(N_{(\alpha, U)}\cap q^{-1}(V))$.
    Let $x\in
    \bigcup_{(\alpha, U)\in S_Q}
    \alpha^{-1}(N_{(\alpha, U)}\cap q^{-1}(V))$.
    Then there is some $\alpha$ so that
    $\alpha(x)\in N_{(\alpha, U)}\cap q^{-1}(V)$ thus
    $[\alpha(x)]\in V$. Since $\imath_{Q}$ is independent of the
    choice of chart, $\alpha_x$,
    $\imath_{Q}(x)=[\alpha(x)]$. Thus
    $x=\imath_{Q}^{-1}([\alpha(x)])\in\imath_{Q}^{-1}(V)$.
    Hence the claim holds and $\imath_{Q}$ is continuous.
    
    I now show that $\imath_{Q}$ is a homeomorphism
    onto its image. Let $V\subset M$ be open. 
    In order for $\imath_{Q}$ to be a homeomorphism onto
    its image it is sufficient to show that $\imath_{Q}(V)$ is open.
    By definition this requires that $q^{-1}(\imath_{Q}(V))$ is open.
    I claim that
    \[
      q^{-1}(\imath_{Q}(V))=\bigcup_{(\alpha, U)\in S_Q}
      \alpha(\dom(\alpha)\cap V).
    \]
    If the claim holds, then as
    $\alpha(\dom(\alpha)\cap V)$ is open $q^{-1}(\imath_{Q}(V))$
    must be open. Hence $\imath_{Q}$ would be an open map.

    I now prove the claim. Let $x\in q^{-1}(\imath_{Q}(V))$.
    Then $x\in N_{(\alpha, U)}$ for some 
    $(\alpha, U)\in S_Q$.
    Thus either $x\in\ran(\alpha)$ or $x\in\partial_{U}\ran(\alpha)$.
    By assumption $q(x)=[x]\in\imath_{Q}(V)$ and therefore
    $\imath_{Q}^{-1}([x])\in V$. Since $V\subset M$ the definition
    of the equivalence relation defining $q$ implies that
    $x\in\ran(\alpha)$. Let $y\in\dom(\alpha)$ be such that
    $\alpha(y)=x$. Then $q(x)=[\alpha(y)]\in\imath_{Q}(V)$. Since
    $\imath_{Q}$ is independent of the choice of $\alpha_y$
    this implies that $y=\imath_{Q}^{-1}([\alpha(y)])\in V$.
    Thus $y\in\dom(\alpha)\cap V$ and so $x\in\alpha(\dom(\alpha)\cap V)$.
    That is, $
    q^{-1}(\imath_{Q}(V))\subset\bigcup_{(\alpha, U)\in S_Q}
      \alpha(\dom(\alpha)\cap V)$.
    Let $x\in\alpha(\dom(\alpha)\cap V)$ for some 
    $(\alpha, U)\in S_Q$,
    then $\alpha^{-1}(x)\in V$. Since $\imath_{Q}$ is independent of the
    choice of $\alpha_x$, $\imath_{Q}(\alpha^{-1}(x))=[x]\in\imath_{Q}(V)$.
    Thus $x\in q^{-1}(\imath_{Q}(V))$ and
    hence the claim holds.

    Since $\imath_{Q}$ is continuous it is clear that $\imath_{Q}(M)$ is open.
    It remains to prove that $\overline{\imath_{Q}(M})=Q(M)$.
    Let $[x]\in Q(M)$ so that $x\not\in\imath_{Q}(M)$.
    By definition there exists a pair 
    $(\alpha, U)\in S_Q$
    so that $x\in N_{(\alpha, U)}$. If $x\in\ran(\alpha)$ then
    $\imath_{Q}(\alpha^{-1}(x))=[x]\in\imath_{Q}(M)$ which is a contradiction.
    Therefore $x\in\partial_{U}\ran(\alpha)$.
    Let $V\subset Q(M)$ be an open neighbourhood
    of $[x]$. Then $q^{-1}(V)\cap N_{(\alpha, U)}$ is open in
    $N_{(\alpha, U)}$
    and hence intersects $\ran(\alpha)$. Let $y\in\ran(\alpha)\cap
    q^{-1}(V)$. Then $\imath_{Q}(\alpha^{-1}(y))=q(y)\in V$.
    Hence every open neighbourhood of $x$ in $Q(M)$
    intersects $\imath_{Q}(M)$. Therefore $\overline{\imath_{Q}(M)}=Q(M)$
    as required.
  \end{proof}

  \begin{corollary}
    The space $Q(M)$ is separable.
  \end{corollary}
  \begin{proof}
    The manifold $M$ is separable and $\overline{\imath_{Q}(M)}=Q(M)$.
    Therefore $Q(M)$ is separable.
  \end{proof}

\section{Properties of $Q(M)$ with respect to
  a compatible set of extensions}\label{sec.compcompexprop}

  This section contains the proofs of
  Proposition \ref{prop.Q(M)fromphi}
  and Corollary \ref{cor.Q(M)isamanifoldifenvboundaryman} which
  are restated here as Proposition \ref{app.prop.Q(M)fromphi}
  and Corollary \ref{app.cor.Q(M)isamanifoldifenvboundaryman}.

  \begin{proposition}[{Proposition \ref{prop.Q(M)fromphi}}]
    \label{app.prop.Q(M)fromphi}
    Let $\phi:M\to M_\phi$ be an envelopment.
    Then, in the notation of Proposition \ref{prop.inverseenvtoext},
    $
      Q=\{(\alpha\circ\phi,\ran(\alpha)):\alpha\in\Atlas{M_\phi},\,
        \dom(\alpha)\cap\partial\phi(M)\neq\varnothing\}
    $
    is a pairwise compatible set of extensions and
    there exists a homeomorphism, in the notation
    of Definition \ref{def.quotientmap},
    $f:Q(M)\to\overline{\phi(M)}$       
    so that $f\circ\imath_Q=\phi$,
    $f(q(N_{(\alpha\circ\phi)}))=\dom(\alpha)\cap\overline{\phi(M)}$
    and $\overline{\beta\circ\phi\circ\phi^{-1}\circ\alpha^{-1}}=
    \beta\circ\alpha^{-1}$.
  \end{proposition}
  \begin{proof}
    I use the notation of Definition \ref{def.quotientmap} throughout 
    this proof.

    Let $(\alpha\circ\phi, \ran(\alpha)),
    (\beta\circ\phi, \ran(\beta))\in Q$, $x\in \BP{\alpha}\cap \ran(\alpha)$
    and $y\in\BP{\beta}\cap \ran(\beta)$ such that
    $(\alpha, U, \{x\})\perp(\beta, V, \{y\})$.
    Then there exists $(z_i)\subset M$
    so that $\alpha\circ\phi(z_i)\to x$ and $\beta\circ\phi(z_i)\to y$.
    Therefore $\phi(z_i)\to \alpha^{-1}(x)$ and
    $\phi(z_i)\to\beta^{-1}(y)$. Since $M_\phi$ is Hausdorff
    $\alpha^{-1}(x)=\beta^{-1}(y)$ which implies that
    $(\alpha\circ\phi,\ran(\alpha),\{x\})\equiv
    (\beta\circ\phi,\ran(\beta),\{y\})$. Thus $Q$ consists
    of pairwise compatible extensions.

    Define $\hat{f}:N_Q\to\overline{\phi(M)}$ by
    \[
      \hat{f}(x)=\left\{\begin{aligned}
        \phi\circ\alpha^{-1}(x) && 
          x\in N_{(\alpha, U)},\, (\alpha, U)\in P_Q\\
        \alpha^{-1}(x) && 
          x\in N_{(\alpha\circ\phi,\ran(\phi))},\,
          (\alpha\circ\phi, \ran(\alpha))\in Q.\\
          \end{aligned}
       \right.
    \]
    It is clear that $\hat{f}$ is well defined. 
    
    I now show that $\hat{f}$ descends to $Q(M)$. 
    Let $x,y\in N_Q$ so that
    $[x]=[y]$, I need to show that $\hat{f}(x)=\hat{f}(y)$.
    We have three cases to check. 
    \textbf{Case 1.} Let $x\in N_{(\alpha, U)}$
    with $(\alpha, U)\in P_Q$ and $y\in N_{(\beta, V)}$
    with $(\beta, V))\in P_Q$. Since $[x]=[y]$, by construction,
    $\beta\circ\alpha^{-1}(x)=y$
    Thus $\hat{f}(x)=\hat{f}(y)$.
    \textbf{Case 2.} Let $x\in N_{(\alpha, U)}$
    with $(\alpha, U)\in P_Q$ and $y\in N_{(\beta\circ\phi, \ran(\beta))}$
    with $(\beta\circ\phi, \ran(\beta)))\in Q$. Since $[x]=[y]$,
    by construction,
    $\beta\circ\phi\circ\alpha^{-1}(x)=y$. Thus
    $\hat{f}(x)=\phi\circ\alpha^{-1}(x)=\beta^{-1}(y)=\hat{f}(y)$.
    \textbf{Case 3.} Let $x\in N_{(\alpha\circ\phi, \ran(\alpha))}$
    with $(\alpha\circ\phi, \ran(\alpha))\in Q$ 
    and $y\in N_{(\beta\circ\phi, \ran(\beta))}$
    with $(\beta\circ\phi, \ran(\beta)))\in Q$. If
    $x\in\dom(\alpha\circ\phi)$
    then as $[x]=[y]$, by construction,
    $\beta\circ\phi\circ\phi^{-1}\circ\alpha^{-1}(x)=y$
    and thus $\hat{f}(x)=\alpha^{-1}(x)=\beta^{-1}(y)=\hat{f}(y)$.
    If $x\not\in\dom(\alpha\circ\phi)$
    then $x\in \partial_{\ran(\alpha)}\ran(\alpha\circ\phi)$.
    Choose $(x_i)\subset\ran(\alpha\circ\phi)$ so that
    $x_i\to x$. Since $[x]=[y]$ there exists a subsequence
    $(y_i)\subset(\phi^{-1}\circ\alpha^{-1}(x_i))$
    so that $(y_i)\subset\dom(\beta\circ\phi)$ and $\beta\circ\phi(y_i)\to
    y$. Thus there exists a sequence $(\alpha^{-1}(x_i))$ of $M_\phi$
    so that $\alpha^{-1}(x_i)\to \alpha^{-1}(x)$ 
    and so that $\beta^{-1}(y)$ is a limit point of $(\alpha^{-1}(x_i))$.
    Since $M_\phi$ is Hausdorff this implies that
    $\alpha^{-1}(x)=\beta^{-1}(y)$. Thus $\hat{f}(x)=\hat{f}(y)$.
    Hence $\hat{f}$ descends to a function $f:Q(M)\to\overline{\phi(M)}$,
    so that $f\circ q = \hat{f}$.

    It is clear that $\hat{f}$ is continuous, thus $f$ is continuous.
    To check that $f$ is open it suffices to prove that
    $f(q(V))$ is open for any $V\subset N_{(\alpha, U)}$ an open subset with
    $(\alpha, U)\in S_Q$. But $f\circ q|_{N_{(\alpha, U)}}=
    \hat{f}|_{N_{(\alpha, U)}}$ which is
    either $\phi\circ\alpha^{-1}$ or $\alpha^{-1}$ both of which are open
    maps. Therefore $f(q(V))$ is open in the subspace topology
    of $\hat{f}(N_{(\alpha, U)})$. Since $\hat{f}(N_{(\alpha, U)})$ is open
    in $\overline{\phi(M)}$ the set $f(q(V))$ is open in 
    $\overline{\phi(M)}$. Hence $f$ is open.
    Suppose that $f([x])=f([y])$, where $x\in N_{(\alpha,U)}$
    and $y\in N_{(\beta, V)}$. Then there exists
    $(z_i)\subset \phi(M)$ so that $z_i\to f([x])$
    and $z_i\to f([y])$.
    By construction of $f$, $f\circ q=\hat{f}$. Thus
    as both $f$ and $q$
    are open $\hat{f}$ is open. Since $\hat{f}$ is onto
    there exists $(x_i)\subset \dom(\alpha)$ 
    and $(y_i)\subset \dom(\beta)$ so that
    $x_i\to x$,
    $y_i\to y$, $\hat{f}(x_i)=\hat{f}(y_i)$ and
    $(\hat{f}(x_i))\subset(z_i)$.
    This is sufficient to show that $(\alpha, U, \{x\})\perp(\beta, V, \{y\})$.
    By assumption
    on $Q$ this implies that $[x]=[y]$ so that $f$ is injective.
    Let $x\in \overline{\phi(M)}$. If there exists $y\in M$
    so that $\phi(y)=x$ then there exists $\alpha\in \Atlas{M}$
    so that $\alpha(y)\in N_{(\alpha,\ran(\alpha))}\subset N_Q$.
    Thus $f(y)=\phi\circ\alpha^{-1}(\alpha(y))=x$.
    Otherwise
    there exists $\alpha\in\Atlas{M_\phi}$
    so that $x\in\dom(\alpha)$. Thus $(\alpha\circ\phi, \ran(\alpha))\in Q$
    and $\alpha(x)\in N_Q$. By construction
    $f([\alpha(x)])=\alpha^{-1}\circ\alpha(x)=x$.
    Hence $f$ is bijective.

    I now show that $f\circ\imath_Q=\phi$. Let $x\in M$ then 
    there
    exists $\alpha\in\Atlas{M}$ so that $x\in\dom(\alpha)$.
    Hence
    $f\circ\imath_Q(x)=f(q(\alpha(x)))=\hat{f}(\alpha(x))
      =\phi\circ\alpha^{-1}(\alpha(x))=\phi(x)$.
    
    I now show that
    $f(q(N_{(\alpha\circ\phi,\ran(\alpha))}))=
    \dom(\alpha)\cap\overline{\phi(M)}$
    for $(\alpha\circ\phi,\ran(\alpha))\in Q$. By definition
    $f(q(N_{(\alpha\circ\phi,\ran(\alpha))}))
      =\alpha^{-1}(N_{(\alpha\circ\phi,\ran(\alpha))})$.
    But, by definition,
    $N_{(\alpha\circ\phi,\ran(\alpha))}=\alpha(\dom(\alpha)\cap\phi(M))\cup
      \partial_{\ran(\alpha)}\alpha(\dom(\alpha)\cap\phi(M))$.
    Since $\ran(\alpha)$ and $\alpha(\dom(\alpha)\cap\phi(M))$
    are open
    $\partial_{\ran(\alpha)}\alpha(\dom(\alpha)\cap\phi(M))
    =\ran(\alpha)\cap\partial\alpha(\dom(\alpha)\cap\phi(M))$.
    Thus to give the set equivalence it is necessary and sufficient
    to show that
    $\alpha(\dom(\alpha)\cap\partial\phi(M))=
    \ran(\alpha)\cap \partial\alpha(\dom(\alpha)\cap\phi(M))$.
    This set equality follows directly from the definitions.

    Lastly, I show that $\overline{\beta\circ\phi\circ\phi^{-1}\circ\alpha}=
    \beta\circ\alpha^{-1}$. Suppose that
    $\overline{\beta\circ\phi\circ\phi^{-1}\circ\alpha}(x)=y$,
    where $x\in N_{(\alpha\circ\phi, \ran(\alpha))}$ and 
    $y\in N_{(\beta\circ\phi, \ran(\beta))}$.
    By definition $q(x)=q(y)$.     
    Thus $\alpha^{-1}(x)=\hat{f}(x)=f(q(x))=f(q(y))=\hat{f}(y)=\beta^{-1}(y)$
    and so $\beta\circ\alpha^{-1}(x)$.
    The reverse implication,
    $\beta\circ\alpha^{-1}(x)=y$ implies $q(x)=q(y)$, holds by definition
    of $Q(M)$. Hence
    $\overline{\beta\circ\phi\circ\phi^{-1}\circ\alpha^{-1}}(x)=y$,
    as required.
  \end{proof}

  \begin{corollary}[{Corollary \ref{cor.Q(M)isamanifoldifenvboundaryman}}]
    \label{app.cor.Q(M)isamanifoldifenvboundaryman}
    Let $\phi:M\to M_\phi$ be an envelopment and let $Q(M)$
    be as given in Definition \ref{def.quotientmap}.
    If $\overline{\phi(M)}$ is a manifold
    with boundary then $Q(M)$ is a manifold with boundary and
    $f$ is a diffeomorphism, where $f$ is given in Proposition 
    \ref{prop.Q(M)fromphi}.
  \end{corollary}
  \begin{proof}
    [{Proof of Corollary \ref{cor.Q(M)isamanifoldifenvboundaryman}}]
    I use the notation of Proposition \ref{prop.Q(M)fromphi}.
    Since $f$ is a homeomorphism $Q(M)$ is second countable and
    Hausdorff as $\overline{\phi(M)}$ is second countable and Hausdorff.
    Let $\Atlas{Q(M)}=\{\beta\circ f:\beta\in\Atlas{\overline{\phi(M)}}\}$.
    The transition maps between charts are
    $\beta\circ f\circ f^{-1}\circ\alpha^{-1}=\beta\circ\alpha^{-1}$ and
    therefore $Q(M)$ is a manifold with boundary. Note that this
    atlas is essentially that given by the extensions in $Q$
    due to the definition of $f$ in Proposition \ref{prop.Q(M)fromphi}.
    In order for $f$ to be a diffeomorphism it must be the case that
    all $\alpha\in\Atlas{Q(M)}$ and all $\beta\in\Atlas{\overline{\phi(M)}}$,
    $\beta\circ f\circ\alpha^{-1}$ must be a diffeomorphism.
    By definition there exists $\gamma\in \Atlas{\overline{\phi(M)}}$
    so that $\alpha=\gamma\circ f$.
    Thus
    $\beta\circ f\circ\alpha^{-1}=\beta\circ f\circ
    f^{-1}\circ\gamma=\beta\circ\gamma$ which is a diffeomorphism
    by assumption on $\overline{\phi(M)}$.
  \end{proof}

\section{Acknowledgements}
The author was supported by Marsden grant UOO-09-022. He would like to 
thank J\"org Frauendiener, Florian Beyer and J\"org Hennig for 
helpful comments during review
of the paper.

\end{document}